\documentclass[abstracton,letterpaper,11pt]{scrartcl}

\usepackage[utf8]{inputenc}
\usepackage[T1]{fontenc}
\usepackage[margin=1in,letterpaper]{geometry}

\usepackage{hyperref}
\usepackage{graphicx}
\usepackage{float}
\newfloat{float}{h}{flt}
\usepackage[usenames]{xcolor}
\usepackage{amsmath}
\usepackage{amssymb}
\usepackage{amsfonts}
\usepackage{amsthm}
\usepackage{xspace}
\usepackage[style=base,font=footnotesize,labelfont=bf,skip=2ex]{caption}
\usepackage{subcaption}
\usepackage{enumitem}

\usepackage{calc}

\usepackage{tabularx}
\usepackage{multirow}

\usepackage{onlyamsmath}
\usepackage{comment}

\usepackage[boxed]{algorithm2e}
\SetAlCapFnt{\footnotesize}
\SetAlCapNameFnt{\footnotesize}
\SetAlCapSkip{1ex}
\SetAlFnt{\small}
\newcommand\MoveEqLeft[1][2]{\kern #1em  &   \kern -#1em}

\let\originalleft\left
\let\originalright\right
\renewcommand{\left}{\mathopen{}\mathclose\bgroup\originalleft}
\renewcommand{\right}{\aftergroup\egroup\originalright}

\newcommand{\Exp}[1]{\ensuremath{\exp\left(#1\right)}}

\usepackage{nicefrac}
\usepackage{xstring}
\newcommand\inlinefrac[2]{#1/#2}
\newcommand\ifrac[2]{\IfInteger{#1}{\IfInteger{#2}{\nicefrac{#1}{#2}}{\inlinefrac{#1}{#2}}}{\inlinefrac{#1}{#2}}}

\hypersetup{%
    colorlinks=true, linktocpage=true, pdfstartpage=1, pdfstartview=FitV,%
    breaklinks=true, pdfpagemode=UseNone, pageanchor=true, pdfpagemode=UseOutlines,%
    plainpages=false, bookmarksnumbered, bookmarksopen=true, bookmarksopenlevel=1,%
    hypertexnames=true, pdfhighlight=/O,%
    urlcolor=darkred, linkcolor=lightblue, citecolor=darkgreen, %
    pdftitle={},%
    pdfauthor={},%
    pdfsubject={},%
    pdfkeywords={},%
    pdfcreator={pdfLaTeX},%
    pdfproducer={LaTeX}%
}

\addtokomafont{disposition}{\rmfamily}

\definecolor{darkred}{rgb}{0.5,0,0}
\definecolor{lightblue}{rgb}{0,0.4,0.8}
\definecolor{darkgreen}{rgb}{0,0.5,0}

\usepackage{aliascnt}
\def\NewTheorem#1#2{%
  \newaliascnt{#1}{theorem}
  \newtheorem{#1}[#1]{#2}
  \aliascntresetthe{#1}
  \expandafter\def\csname #1autorefname\endcsname{#2}
}

 \newtheorem{theorem}{Theorem}[section]
\NewTheorem{lemma}{Lemma}
\NewTheorem{corollary}{Corollary}
\NewTheorem{conjecture}{Conjecture}
\NewTheorem{observation}{Observation}
\NewTheorem{definition}{Definition}
\NewTheorem{proposition}{Proposition}
\NewTheorem{remark}{Remark}

\NewTheorem{claim}{Claim}

\newtheoremstyle{restate}{\topsep}{\topsep}{\itshape}{0pt}{\bfseries}{.}{5pt plus 1pt minus 1pt}{\thmname{#1}\thmnumber{ \begin{NoHyper}\ref{#3}\end{NoHyper}}}
\theoremstyle{restate}
\newtheorem{restate}{Theorem}

\renewcommand*{\algorithmautorefname}{Algorithm}%

\theoremstyle{remark}

\newcommand{\ceil}[1]{\ensuremath{\left\lceil#1\right\rceil}}
\newcommand{\floor}[1]{\ensuremath{\left\lfloor#1\right\rfloor}}

\DeclareMathOperator{\BIGO}{O}
\DeclareMathOperator{\LITTLEO}{o}
\DeclareMathOperator{\BIGOMEGA}{\Omega}
\DeclareMathOperator{\LITTLEOMEGA}{\omega}
\DeclareMathOperator{\BIGTHETA}{\Theta}

\newcommand{\BigO}[1]{\ensuremath{\BIGO\left(#1\right)}}
\newcommand{\LittleO}[1]{\ensuremath{\LITTLEO\left(#1\right)}}
\newcommand{\BigOmega}[1]{\ensuremath{\BIGOMEGA\left(#1\right)}}
\newcommand{\LittleOmega}[1]{\ensuremath{\LITTLEOMEGA\left(#1\right)}}
\newcommand{\BigTheta}[1]{\ensuremath{\BIGTHETA\left(#1\right)}}

\newcommand{\onebit}{{\normalfont\texttt{OneExtraBit}}\xspace}

\newcommand\whp{with high probability\whpFootnote\xspace}

\newcommand\Whp{With high probability\whpFootnote\xspace}

\newcommand\whpFootnoteText{\footnote{Throughout this paper, the expression \emph{\whp} means a probability of at least $1 - n^{-\BigOmega{1}}$.}}

\newcommand\whpFootnote{\global\def\whpFootnote{}\whpFootnoteText}

\DeclareMathOperator{\PROBABILITY}{P\scriptstyle{r}}
\newcommand{\Probability}[1]{\ensuremath{ \PROBABILITY\left[#1\right] }}

\DeclareMathOperator{\EXPECTED}{E}
\newcommand{\Expected}[1]{\ensuremath{ \EXPECTED\left[#1\right] }}

\DeclareMathOperator{\VARIANCE}{Var}
\newcommand{\Variance}[1]{\ensuremath{ \VARIANCE\left[#1\right] }}

\newcommand\numberthis{\addtocounter{equation}{1}\tag{\theequation}}

\def\A{\ensuremath{\mathcal{A}}\xspace}
\def\B{\ensuremath{\mathcal{B}}\xspace}
\def\C#1{\ensuremath{\mathcal{C}_{#1}}\xspace}

\def\Ci{\C{i}}
\def\notC#1{\ensuremath{\overline{\mathcal{C}_{#1}}}\xspace}

\def\ttrue{\ensuremath{\text{\scriptsize{\textsc{True}}}}\xspace}
\def\tfalse{\ensuremath{\text{\scriptsize{\textsc{False}}}}\xspace}
\def\tnull{\ensuremath{\text{\scriptsize{\textsc{Null}}}}\xspace}

\newcommand\eps[1]{\ensuremath{\varepsilon_{\scriptscriptstyle #1}}}
\newcommand\teps[1]{\ensuremath{\varepsilon_{\scriptscriptstyle \text{\normalfont #1}}}}

\def\P#1{#1^{\scriptscriptstyle(P)}}
\def\PP#1{#1^{\scriptscriptstyle(P')}}

\def\lefttag#1{\tag*{\makebox[0pt][l]{\hspace*{-\textwidth}#1}}}

\newcommand\SG{Sync Gadget\xspace}
\newcommand\BP{Bit-Propagation\xspace}
\newcommand\TC{Two-Choices\xspace}

\usepackage{placeins}

\usepackage{dcmtitle}

\title{Rapid Asynchronous Plurality Consensus}

\author[1]{Robert Elsässer}
\author[2]{Tom Friedetzky}
\author[1]{Dominik Kaaser}
\author[3]{Frederik Mallmann-Trenn}
\author[1]{Horst Trinker}

\affiliation{University of Salzburg, Austria\\
\textit{elsa@cosy.sbg.ac.at},\, \textit{dominik@cosy.sbg.ac.at},\, \textit{horst.trinker@sbg.ac.at}}

\affiliation{Durham University, U.K.\\
\textit{tom.friedetzky@dur.ac.uk}}

\affiliation{École normale supérieure, Paris, France\\
Simon Fraser University, Canada\\
\textit{mallmann@di.ens.fr}
}
\begin{document}
\maketitle
\thispagestyle{empty}

\begin{abstract}
We consider distributed plurality consensus in a complete graph of size $n$
with $k$ initial opinions. We design an efficient and simple protocol in the
asynchronous communication model that ensures that all nodes eventually agree
on the initially most frequent opinion. In this model, each node is equipped
with a random Poisson clock with parameter $\lambda=1$. Whenever a node's clock
ticks, it samples some neighbors, uniformly at random and with replacement, and
adjusts its opinion according to the sample.

Distributed plurality consensus has been deeply studied in the synchronous
communication model, where in each round, every node chooses a sample of its
neighbors, and revises its opinion according to the obtained sample. A
prominent example is the so-called two-choices algorithm, where in each round,
every node chooses two neighbors uniformly at random, and if the two sampled
opinions coincide, then that opinion is adopted. This protocol is very
efficient and well-studied when $k=2$. If $k=\BigO{n^\varepsilon}$ for some
small $\varepsilon$, we show that it converges to the initial plurality opinion
within $\BigO{k \cdot \log{n}}$ rounds, w.h.p., as long as the initial
difference between the largest and second largest opinion is $\BigOmega{\sqrt{n
\log n}}$. On the other side, we show that there are cases in which
$\BigOmega{k}$ rounds are needed, w.h.p. 

One can beat this lower bound by combining the two-choices protocol with
push-pull broadcasting. The main idea is to divide the process into several
phases, where each phase consists of a two-choices round followed by several
broadcasting rounds. This, however, is difficult to realize in the asynchronous
model, as we can no longer rely on nodes performing the same operations at the
same time.

Our main contribution is just that: a non-trivial adaptation of this approach
to the asynchronous model. If the support of the most frequent opinion is at
least $(1+\varepsilon)$ times that of the second-most frequent one and
$k=\BigO{\Exp{\log{n}/\log \log{n}}}$, then our protocol achieves the best
possible run time of $\BigO{\log n}$, w.h.p. Key to our adaptation is that we
relax full synchronicity by allowing $\LittleO{n}$ nodes to be poorly
synchronized, and the well synchronized nodes are only required to be within a
certain time difference from one another. We enforce this ``sufficient''
synchronicity by introducing a novel gadget into the protocol. Other parts of
the adaptation are made to work using arguments and techniques based on a Pólya
urn model.

\end{abstract}
\vspace{0.1cm}

\clearpage
\tableofcontents
\newpage
\section{Introduction} \label{sect:introduction}
Distributed voting is a fundamental problem in distributed computing with
applications in a multitude of fields. In distributed computing, these include,
among others, consensus \cite{HP01} and leader election \cite{BMPS04}.

We consider the following  plurality consensus process on the clique $K_n$ of
size $n$. Each node in the network starts with one initial opinion, which we
also refer to as color, from a finite set of possible opinions. We distinguish
between the synchronous and the asynchronous setting. In the \emph{synchronous
model}, all nodes communicate simultaneously with some of their neighbors and
update their opinions accordingly. In the \emph{asynchronous model}, we assume
that each node has a random clock which ticks according to a Poisson
distribution, once per unit of time in expectation. Again, upon activation a
node updates its opinion according to a sample of its neighborhood.

Regardless of the underlying model of synchronicity, if eventually all nodes
agree on one opinion, we say this opinion \emph{wins}, and the process
\emph{converges}. Typically, one would demand from such a voting procedure to
run accurately, that is, the opinion with the largest initial support should
win with decent probability $(1-\LittleO{1})$, and to be efficient, that is,
the voting process should converge within as few communication steps as
possible. Additionally, voting algorithms are usually required to be simple,
fault-tolerant, and easy to implement \cite{HP01,Joh89}.

\subsection{Model}
In the following section, we will introduce formally the model which we
consider in the remainder of this paper. We give a formal definition of the
consensus process in the synchronous and the asynchronous model followed by an
overview of our results in \autoref{sect:our-contribution}.

We consider the following plurality consensus process on the clique $G=(V,E) =
K_n$ of size $n$. Initially, the nodes are partitioned into $k$ groups
representing $k$ colors $\mathcal{C}_1, \dots, \mathcal{C}_k$. We will denote
the number of nodes having color \C{j} as $c_j$. We furthermore denote the set
of all colors as $C = \left\{\mathcal{C}_1, \dots, \mathcal{C}_k\right\}$.
Also, we will occasionally abuse notation and use $\mathcal{C}_i$ to denote the
\emph{set} of all vertices having color $\mathcal{C}_i$. W.l.o.g., we assume
that colors are ordered in descending order such that $c_1 \geq c_2 \geq \dots
\geq c_k$. We will denote the initial plurality color \C1 as \A with size $a =
c_1$ and we will use \B to denote the second largest color \C{2} of size $b =
c_2$.

\subsubsection*{Synchronous Model}
In the synchronous model we assume that the protocol operates in discrete
rounds. In each round, the nodes may simultaneously sample other nodes
uniformly at random and then simultaneously change their opinion as a function
of the observed samples. One prominent example here is the \TC process where in
each round every node samples two nodes chosen uniformly at random, with
replacement. If the chosen nodes' colors coincide, then the node adopts this
color. We denote this process as the \emph{plurality consensus process with two
choices}. Our first two results will be shown w.r.t.\ this synchronous model.

\subsubsection*{Parallel Asynchronous Model}
In the asynchronous model, every node $v$ is equipped with a random clock which
ticks according to a Poisson distribution with parameter $\lambda = 1$.
Whenever a node ticks, it may sample nodes chosen uniformly at random and
update its opinion based on the sampled values. That is, we assume a
memory-less random clock, such that for every node the time between two ticks
is exponentially distributed with parameter $\lambda = 1$. Consequently, from
the memory-less property it follows that at any time $t$ each node has the same
probability $1/n$ to be the next one to tick.

\subsubsection*{Sequential Asynchronous Model}
While the parallel model described above represents real-world processes for
which event frequencies are commonly modeled by Poisson clocks, we give in the
following a more theoretical yet equivalent model.

The Poisson distribution used for the clocks in the parallel model has the
so-called \emph{memory-less property}. That is, at any given time $t$,
regardless of the previous events, every node has exactly the same probability
to be the next node to tick, namely $1/n$. We furthermore assume that, upon a
node's activation, the execution of one step occurs atomically, that is, no two
nodes are ever active concurrently. Therefore, instead of considering the
asynchronous parallel process in continuous time, we rather analyze the process
in the so-called sequential model. In this sequential model, we assume that a
discrete time is given by the sequence of ticks, and at any of the discrete
time steps, a node is selected to perform its task uniformly at random from the
set of all nodes.

Observe that we can relate the number of ticks in the sequential model to the
continuous time in the asynchronous model as follows (see also \cite{AGV15}).
We have for any tick $t$ in the asynchronous sequential model that
$\Expected{T_t} = t/n$, where $T_t$ is the random variable for the continuous
time of tick $t$. Moreover, for the expected number of ticks allotted by the
asynchronous voting algorithm described in
\autoref{sect:asynchronous-analysis}, we obtain that the continuous time is
concentrated around the expected value such that \whp the asynchronous voting
process converges after at most $\BigO{\log{n}}$ time units. See, e.g.,
\cite[Lemma~1]{BGPS06} for details on the concentration. 

\subsection*{Stability}
In our analysis, we will show that the \TC process can tolerate the presence of
an adversary which is allowed to arbitrarily change the opinion of up to
$F=c_1(c_1-c_2)/(8n)$ arbitrarily selected nodes after every round. We will
show that under these assumptions our \TC process still guarantees  that \whp a
vast majority of nodes accept the plurality opinion, that is, the initially
dominant opinion. Observe that, similarly, all our theorems also hold if the
adversary is allowed to change opinions at the \emph{beginning} of a round. We
use a definition similar to the definition by Becchetti et al.\ \cite{BCNPT15},
which in turn has its roots in \cite{AAE08, AFJ06}.

\begin{definition} \label{def:stability}
A \emph{stabilizing near-plurality protocol} ensures the following properties:
\begin{enumerate}
\item \emph{Almost agreement.} Starting from any initial configuration, in a
finite number of rounds, the system must reach a regime of configurations where
all but a negligible \emph{bad} subset of nodes of size at most
$\BigO{n^\varepsilon}$ for some constant $\varepsilon < 1$ support the same
opinion.
\item \emph{Almost validity.} Given a large enough initial bias, the system is
required to converge to the plurality  opinion $\A$, \whp, where all but a
negligible \emph{bad} set of nodes have opinion $\A$.
\item \emph{Non-termination.} In dynamic distributed systems, nodes represent
simple and anonymous computing units which are not necessarily able to detect
any global property.
\item \emph{Stability.} The convergence to such a weaker form of agreement is
only guaranteed to hold \whp.
\end{enumerate}
\end{definition}

\subsection{Our Contributions} \label{sect:our-contribution}

In this paper we consider a modification of the \TC protocol to design an
efficient distributed voting algorithm, allowing for a large number of
different opinions in the asynchronous settings. So far, most work in this area
concentrated on the synchronous communication model. As we see below, the \TC
protocol has certain limitations -- even in this synchronous setting.

\paragraph{Limits of the Two-Choices Approach.}

The \TC protocol seems to be very efficient if the number of colors is two
\cite{CER14}. The following result can be seen as an extension of Cooper et
al.\ \cite{CER14} on the complete graph when initially the number of opinions
is larger than two. That is, we assume that every node of the clique $G = (V,E)
= K_n$ has one of $k$ possible opinions at the beginning, where $k =
\BigO{n^\epsilon}$ for some small positive constant $\epsilon$. Then, the
following theorem holds.

\begin{theorem} \label{thm:main-result}
Consider the synchronous model. Let $G=K_{n}$ be the complete graph with $n$
nodes. Let $k=\BigO{n^\varepsilon}$ be the number of opinions for some small
constant $\varepsilon > 0$. The \TC plurality consensus process defined in
\autoref{alg:two-choices} converges \whp to \A within $\BigO{{n}/{c_1} \cdot
\log{n}}$ rounds, if the initial bias is at least $c_1-c_2 \geq z \cdot
\sqrt{n\log{n}}$ for some constant $z$. Assuming this bias, the process
fulfills the stabilizing near-plurality conditions in presence of any
$F=c_1(c_1-c_2)/(8n)$-dynamic adversary.

Furthermore, if we assume that $c_1-c_2 = z \cdot \sqrt{n\log{n}}$ for some
constant $z$, and $c_j =c_2$ for any $j = 3, \dots , k$, then the \TC protocol
requires $\Omega(n/c_1+\log n)$ rounds in expectation to converge.
\end{theorem}

The difficulty in the analysis lies in the possibly diminishingly small initial
\emph{mass} of $\A$ in comparison to the mass of all other colors.
Interestingly, the required initial gap does not depend on the number of
opinions present. Moreover, we also show that if $c_1-c_2 = \BigO{\sqrt{n}}$,
then $\B$ wins with constant probability.

Slightly later (cf.~\cite{CRRS16,EFKMT16}), Cooper et al.~proved the same run
time in a much more general form by considering the class of regular expander
graphs, albeit assuming a slightly more restrictive initial bias. 

In order to overcome the $\BigOmega{k}$ lower bound in general, we need to
modify the \TC protocol. 

\paragraph{Breaking the $\mathbf{\Omega(k)}$ Barrier in the Synchronous Setting.}

To achieve a low run time, we combine the two choices process with a rumor
spreading algorithm. We first consider this approach in the synchronous setting
and denote the corresponding algorithm by \onebit. For this, we investigate a
slightly modified model called the \emph{memory model}, which is described in
full detail in \autoref{sect:memory}. In this model, we allow each node to
transmit one additional bit. As stated in \autoref{thm:memory}, this allows us
to reduce the run time from $\BigO{{n}/{c_1} \cdot \log{n}}$ to
$\BigO{\left(\log(c_1/\left(c_1-c_2\right))+\log\log{n}\right)\cdot\left(\log{k}
+ \log\log{n}\right)}=\BigO{\log^2 n}$, and the dominating color still wins
\whp, while the initial bias needs only to be slightly larger than in
\autoref{thm:main-result}. If we assume that a tight upper bound on ${n}/{c_1}$
is known to the nodes, the run time of \onebit can further be improved to
$\BigO{\left(\log\log{n}\right)\cdot\left(\log({{n}/{c_1}}) +
\log\log{n}\right)}$. The theorem is formally stated as follows.

\def\theoremmemory{
Consider the synchronous model. Let $G=K_{n}$ be the complete graph with $n$
nodes. Let $k=\BigO{n^\varepsilon}$ be the number of opinions for some small
constant $\varepsilon > 0$. Assume $c_1-c_2 \geq z \cdot \sqrt{n\log^3{n}}$ for
some constant~$z$, then the plurality consensus process \onebit defined in
\autoref{alg:memory} on $G$ converges within 
\begin{equation*}
\BigO{\left(\log(c_1/\left(c_1-c_2\right))+\log\log{n}\right)\cdot\left(\log{k} + \log\log{n}\right)}
\end{equation*}
rounds to \A, \whp.
}
\begin{theorem} \label{thm:memory}
\theoremmemory
\end{theorem}
This can be further improved to
$\BigO{\left(\log(c_1/\left(c_1-c_2\right))+\log\log{n}\right)\cdot\log{k}}$ if
we change the algorithm slightly as described in \autoref{sect:memory}. Coming
from a different angle, essentially the same result was obtained independently
by Berenbrink et al.\ \cite{BFGK16} (see their first protocol) as an
intermediate step toward their main result, as well as by Ghaffari and Parter
\cite{GP16}. To obtain our main result, we will generalize this approach to the
asynchronous communication model.

Note that in the classical \TC protocol each node is implicitly assumed to have
local memory of a certain size, which is used, e.g., to store its current
opinion. The main difference between the classical model and the memory model
is that in the memory model each node also transmits an additional bit along
with its opinion when contacted by a neighbor. Also, the nodes need additional
\emph{local} memory to count the number of rounds. The protocol of
\autoref{thm:memory} ensures that the dominant color $\A$ wins within a small
(at most $\BigO{\log^2{n}}$) number of rounds, even if the bias is only
$\BigO{\sqrt{n\log^3 n }}$. The thorough analysis of this synchronous algorithm
is the basis for understanding and analyzing the corresponding asynchronous
protocol.

\paragraph{Our Main Contribution.}
Our main contribution is an adaptation of the algorithm \onebit to the
asynchronous setting. The main question is whether the same (or similar)
results as in the synchronous case can also be obtained in the asynchronous
setting. As discussed below in more detail, a straight-forward observation is
that in the sequential asynchronous model many nodes may remain
\emph{unselected} for up to $\BigO{\log n}$ time, which implies that no
algorithm can converge in $\LittleO{\log n}$ time. Thus, our aim is to
construct a protocol that solves plurality consensus in $\BigO{\log n}$ time.
We show that if the difference between the numbers of the largest two opinions
is at least $\BigOmega{c_2}$, where $c_2$ is the size of the second largest
opinion, and $k = n^{\BigO{1/\log\log{n}}}$, then our algorithm solves
plurality consensus and achieves the best possible run time of $\BigO{\log n}$,
provided a node is allowed to communicate with at most constantly many other
nodes in a step.

The key to the rapidity of \onebit is that we pair a phase in which all nodes
execute the \TC process with a phase in which successful opinions are
propagated quickly -- much like in broadcasting. For this to work it is crucial
to separate the two phases. While this is trivial in the synchronous setting,
it is impossible in the asynchronous setting. The number of activations of
different nodes can easily differ by $\BigTheta{\log n}$, rendering any attempt
of full synchronization futile if one aims for a run time of $\BigO{\log n}$.
Thus, we restrict ourselves to the concept of \emph{ weak synchronicity } as
follows. At any time we only require that a $(1-\LittleO{1})$-fraction of nodes
are \emph{almost} synchronous. To cope with the influence of the remaining
nodes, we rely on a toolkit of gadgets, which we believe are interesting in
their own right. The obtained weak synchronicity allows us to reuse the
high-level structure of the proof and the analysis of \onebit. Our result is
formally stated in the following theorem.

\def\theoremasync{
Consider the asynchronous model. Let $G=K_{n}$ be the complete graph with $n$
nodes. Let $k= \BigO{\exp\left(\log{n}/\log\log{n}\right)}$ be the number of
opinions. Let $\teps{bias} > 0$ be a constant. Assume $c_1 \geq
\left(1+\teps{bias}\right)\cdot c_i$ for all $i \geq 2$, then the asynchronous
plurality consensus process defined in \autoref{sect:asynchronous-analysis} on
$G$ converges within time $\BigTheta{\log{n}}$ to the majority opinion \A,
\whp.
}
\begin{theorem} \label{thm:async}
\theoremasync
\end{theorem}

\subsection{Related Work}
This overview concentrates on results concerned with Pull Voting, Plurality
Consensus, and Population protocols.

\paragraph{Protocols Based on Pull Voting.}

One major line of research on plurality consensus has its roots in gossiping
and rumor spreading. Communication in these models is often restricted to pull
requests, where nodes can query other nodes' opinions and use a simple rule to
update their own opinion.

In the remainder of this paper we refer to the opinion with initially largest
(second-largest, etc.) support as the largest (second-largest, etc.) opinion.
This does not in any way refer to a possible numerical value that may be
associated with an opinion. One straightforward variant is the so-called
\emph{pull voting} running in discrete rounds, during which each player
contacts a node chosen uniformly at random from among its neighbors and adopts
the opinion of that neighbor. The two papers by Hassin and Peleg \cite{HP01}
and Nakata et al.\ \cite{NIY99} have considered the discrete time two-opinion
voter model on connected graphs. In these papers, each node is initially
assigned one of two possible opinions. Their main result is that the
probability for one opinion to win is proportional to the sum of the degrees of
all vertices supporting that opinion. It has furthermore been shown by Hassin
and Peleg \cite{HP01} that the expected time for the two-opinion voting process
to converge on general graphs can only be bounded by $\BigO{n^3\log{n}}$.
Tighter bounds for general graphs were obtained by \cite{CEOR13,BGKM16,KMS16}.

The expected convergence time for pull voting is at least $\BigOmega{n}$ on
many graphs, such as regular expanders and complete graphs. Taking into account
that solutions to many other fundamental problems in distributed computing,
such as information dissemination \cite{KSSV00} or aggregate computation
\cite{KDG03}, are known to run much more efficiently, Cooper et al.\ noted that
there is room for improvement. To address this issue, Cooper et al.\
\cite{CER14} introduced the \TC voting process. In this modified process, one
is given a graph $G=(V, E)$ where each node has one of two possible opinions.
The process runs in discrete rounds during which, unlike in the classical pull
voting, every node is allowed to contact two neighbors chosen uniformly at
random. If both neighbors have the same opinion, then this opinion is adopted,
otherwise the calling vertex retains its current opinion in this round.

They show that in random $d$-regular graphs, with high probability all nodes
agree after $\BigO{\log n}$ steps on the initially most frequent opinion,
provided that $c_1-c_2=K\cdot(n\sqrt{ 1/d + d/n})$ for $K$ large enough, where
$c_1$ and $c_2$ denote the support of the initially most frequent and
second-most frequent colors. For an arbitrary $d$-regular graph $G$, they need
$c_1-c_2=K\cdot\lambda_2\cdot n$. In the more recent work by Cooper et al.\
\cite{CERRS15}, the results from \cite{CER14} have been extended to general
expander graphs, cutting out the restrictions on the node degrees but
nevertheless proving that the convergence time for the voting procedure remains
in $\BigO{\log{n}}$. Recently, the authors of \cite{CRRS16} showed the
following bound on the consensus time in regular expanders. If the initial bias
between the largest and second-largest opinion is at least $c_1 - c_2 \geq Cn
\max\{\sqrt{\log n/c_1}, \lambda_2 \}$, where $\lambda$ is the absolute second
eigenvalue of the matrix $P = Adj(G)/d$ and $C$ is a suitable constant, then
the largest opinion wins in $\BigO{(n \log n)/c_1}$ steps, with high
probability.

One extension is five-sample voting in $d$-regular graphs with $d \geq 5$,
where in each round at least five distinct neighbors are consulted. Abdullah
and Draief showed an $\BigO{\log_d\log_d{n}}$ bound \cite{AD15}, which is tight
for a wider class of voting protocols. A more general analysis of multi-sample
voting has been conducted by Cruise and Ganesh \cite{CG14} on the complete
graph.

\paragraph{Protocols for Plurality Consensus.}
Becchetti et al.\ \cite{BCNPST14} consider a similar update rule on the clique
for $k$ opinions. Here, each node pulls the opinion of three random neighbors
and adopts the majority opinion among those three (breaking ties uniformly at
random). They need $\BigO{\log k}$ memory bits and prove a tight run time of
$\BigTheta{k\cdot\log n}$ for this protocol, given a sufficiently large bias
$c_1-c_2$. Moreover, they show that if the bias is only of order $\sqrt{kn}$,
then with constant probability the difference $c_1-c_2$ decreases. As we show
in this paper, the \TC process behaves differently since the difference
required by the two choices process is only $\BigOmega{\sqrt{n \log n}}$. The
reason for this phenomenon is that the variance of the number of nodes
switching per round differs greatly between these two processes. In the regime
where all opinions are roughly of the same size, the probability of switching
in the \TC process is $\LittleO{1}$, whereas it is $1-\LittleO{1}$ in the
$3$-majority process. More details can be found in \autoref{sec:3majority}.

In another recent paper, Becchetti et al.\ \cite{BCNPS15} build upon the idea
of the 3-state population protocol by Angluin et al.\ \cite{AAE08}. Using a
slightly different time and communication model, they generalize the protocol
to $k$ opinions. In their model, nodes act in parallel and in each round pull
the opinion of a random neighbor. If it holds for the largest color that
$c_1\geq(1+\varepsilon)\cdot c_2$ for a constant $\varepsilon>0$, the number of
colors is bounded by $k=\BigO{(n/\log n)^{\ifrac{1}{3}}}$, and assuming the
availability of $\log k+\BigO{1}$ bits of memory, their protocol agrees with
high probability on the plurality opinion in time
$\BigO{\operatorname{md}(\mathbf{c})\cdot\log n}$ in the clique. Here,
$\operatorname{md}(\mathbf{c})$ is the so-called \emph{monochromatic distance}
that depends on the initial opinion distribution $\mathbf{c}$. In contrast to
all the results above for $k>2$ opinions, we only require a bias of size
$\BigO{\sqrt{n\log n}}$.

Also interested in balancing the requirement for additional memory with
convergence time, in \cite{BFGK16} the authors propose two plurality consensus
protocols. Both assume a complete graph and realize communication via the
random phone call model. The first protocol is very simple and, \whp, achieves
plurality consensus within $\BigO{\log(k)\cdot\log\log_\gamma n + \log\log n}$
rounds using $\BigTheta{\log\log k}$ bits of additional memory. The second,
more sophisticated protocol achieves plurality consensus within
$\BigO{\log(n)\cdot\log\log_\gamma n}$ rounds using only $4$ overhead bits. In
both cases, $k$ denotes the number of colors, and $\gamma$ denotes the initial
relative plurality gap, the ratio between the plurality opinion and the
second-largest opinion. They require an initial absolute gap of
$\LittleOmega{\sqrt{n}\log^2 n}$. At the heart of their protocols lies the use
of the \emph{undecided state}, originally introduced by Angluin et al.\
\cite{AAE08}. A very recent result by Ghaffari and Parter \cite{GP16}
introduces a protocol for plurality consensus with time and memory bounds
similar to our bounds for \autoref{alg:memory}. They employ a similar basic
idea of consolidation and bit-propagation rounds, which they refer to as
selection and recovery. While aspects of \cite{GP16} and the first protocol in
\cite{BFGK16} are similar to our own protocol (in terms of expectation but not
distribution), they were all developed independently and initially approached
the problem with different specific objectives.

Another interesting model allows for adversarial corruption of opinions. Doerr
et al.\ \cite{DGMSS11} investigate the so-called $3$-median rule which allows
an adversary to arbitrarily change the opinion of $F=\sqrt{n}$ arbitrary
nodes. The required time to reach near-consensus is $\BigO{\log k \log\log n +
\log n}$, where $k$ is the size of the set of opinions. Their algorithm assumes
a total ordering on the opinions and requires nodes to be able to perform basic
algebraic operations. In a recent paper, Becchetti et al.\ \cite{BCNPT15}
overcome these assumptions and show that the $3$-majority rule is stable
against an $F=\LittleO{\sqrt{n}}$ dynamic-adversary. It is worth noting that
both \cite{BCNPT15, DGMSS11} are only interested in consensus and not
necessarily plurality, which would mean that the initially dominant color wins
\whp if the initial bias is large enough.

\paragraph{Population Protocols}
The second major line of work on majority voting considers \emph{population
protocols}, in which the nodes usually act asynchronously. In its basic
variant, nodes are modeled as finite state machines with a small state space.
Communication partners are chosen either adversarially or randomly,
see~\cite{AAER07, AR07} for a more detailed description. Angluin et al.\
\cite{AAE08} propose a 3-state (that is, constant memory) population protocol
for majority voting with $k=2$ in the clique to model the mixing behavior of
molecules. We refer to their communication model as the \emph{sequential
model}. In each time step, an edge is chosen uniformly at random, such that
only one pair of nodes communicates. They show that consensus is reached after
$\BigO{n \log n}$ time steps where the largest opinion has an initial size of
at least $n/2+\omega{\sqrt n\log n}$. To allow for an easier comparison with
the synchronous model, we will normalize the run time of all sequential
algorithms and continuous processes throughout this paper by dividing their run
time by $n$ \cite{AGV15}. To make this explicit, we sometimes refer to this as
\emph{parallel time}. This is a typical measure for population protocols and
based on the intuition that, in expectation, each node communicates with one
neighbor within $n$ time steps. In a recent paper, Alistarh et al.\
\cite{AGV15} gave a sophisticated sequential protocol for $k=2$ in the clique.
It solves exact majority and has, \whp, parallel run time $\BigO{ \log^2 n /
\left(s\cdot\left(c_1-c_2\right)\right)+\log^2 n\cdot\log s}$, where $s$ is the
number of states with $s$ asymptotically in $[\log n\cdot\log\log n, n]$.

\section{Plurality Consensus with Two Choices} \label{sect:main-result}

\begin{algorithm}[t]
\SetAlgoVlined
\SetKwFor{AtNode}{at each node}{do in parallel}{end}
\SetKw{KwLet}{let}

\SetKwProg{Algorithm}{Algorithm}{}{end}
\SetKwFunction{TwoChoices}{two-choices}
\SetKwFunction{Color}{color}
\SetKwFunction{Bit}{bit}

\SetKw{KwLet}{let}
\Algorithm{\TwoChoices{$G = (V, E)$, $ \Color : V \rightarrow C$}}{
\For{round $t = 1$ \KwTo $|C|\cdot\log{|V|}$}{
\AtNode{$v$}{
\KwLet $u_1, u_2 \in N(v) $ uniformly at random\;
\If{\Color{$u_1$} $=$ \Color{$u_2$}}{
    \Color{$v$} $ \gets $ \Color{$u_1$}\;
}
}
}
}
\caption{Distributed Voting Protocol with Two Choices}
\label{alg:two-choices}
\end{algorithm}

In \autoref{sec:TCupper} we show the upper bound \autoref{thm:main-result} on
the \TC process. We show that if the initial bias is $\BigOmega{n \log n}$,
then the initially most dominant color wins \whp in $\BigO{ k\cdot \log n}$
rounds.

In \autoref{sect:lowerbound} we show two lower bounds: We show that if the
initial bias is of order $\BigO{\sqrt{n}}$, then with constant probability a
color different than \A will win (\autoref{thm:lowerbound-bias}). Furthermore,
we show that there are configurations from which we require $\BigOmega{k+\log
n}$ rounds until any opinion wins (\autoref{thm:lowerbound-runtime}).

\subsection{Upper bound}\label{sec:TCupper}  

In this section we show our first  theorem stated in \autoref{thm:main-result}.
The algorithm discussed in this section is formally defined in \autoref{alg:two-choices}.
The structure of the proofs is as follows. 
We show using Chernoff bounds that the number of nodes which change their opinion to $\A$ is 
larger than the number of nodes which switch to $\B$. Given that the initial bias is large enough, the relative difference between $\A$ and $\B$ increases rapidly in every round \whp, and using a union bound yields the theorem.
The difficult part lies in bounding the number of switches to $\A$ and to $\B$. Indeed, just
applying a Chernoff bound to every single color appears to lead to much weaker results. Instead, we carefully aggregate colors when considering the nodes switching to $\A$ or $\B$.
Intuitively, the difficulty lies in the sheer number of initial opinions we allow. In contrast to what is permitted in most previous work, their total mass may significantly exceed the initial mass of \A.

Let
$f_{ij}$ denote the random variable denoting the \emph{flow} from color
$\mathcal{C}_i$ to color $\mathcal{C}_j$, that is, $f_{ij}$ at a given time
step $t$ represents the number of nodes which had color $\mathcal{C}_i$ at the
previous time step $t-1$ and switched to color $\mathcal{C}_j$ at time $t$. 
We
will use $c_1', \dots, c_k'$ to denote the number of nodes of corresponding
colors after the switching has been performed before the adversary changes $F$ arbitrary nodes.

For simplicity of notation, we will assume that in the following the dominating
color \C1 is denoted as \A with $a = c_1$. Furthermore, we will use \B to denote
the second largest color \C{2} of size $b = c_2$. Also, we
will use $f_{\A\B}$ and $f_{\B\A}$ to denote $f_{1,2}$ and $f_{2,1}$, respectively.

Observe that in the complete graph the number $f_{ij}$ of nodes switching from
\C{i} to \C{j} has a binomial distribution with parameters $f_{ij} \sim
B(c_i,\,\ifrac{c_j^2}{n^2})$. Clearly, the expectation and variance of
$f_{ij}$ are
\begin{align*}
\Expected{f_{ij}} = \frac{c_i\cdot c_j^2}{n^2} && \text{and} &&
\Variance{f_{ij}} = \frac{c_i\cdot c_j^2\left(n-c_j\right)\left(n+c_j\right)}{n^4} \enspace .
\end{align*}

Observe that if $a \geq \left(\ifrac{1}{2} + \eps{1}\right)n$ for some constant
$\eps{1} > 0$, the process converges within $\BigO{\log{n}}$ steps \whp.
This follows from \cite{CER14} since in
the case of $a \geq \left(\ifrac{1}{2} + \eps{1}\right)n$ the process is
stochastically dominated by the two color voting process. For the sake of readability we
assume in the following that $a \leq \ifrac{n}{2}$. Furthermore, observe that
$a>\ifrac{n}{k}$, since \A is the largest of $k$ color classes. We start with the
following definitions.

Let $S \subseteq C$ be a set of colors. We will use the random variable
$f_{iS}$ to denote the sum of all flows from color \C{i} to any color in $S$
and $f_{Si}$ to denote the sum of all flows from any color in $S$ to \C{i}. We
have in expectation
\begin{align*}
\Expected{f_{Si}} &= \sum_{\C{j} \in S}\frac{c_j \cdot c_i^2}{n^2}  & \text{and} && \Expected{f_{iS}} &= \sum_{\C{j} \in S}\frac{c_i \cdot c_j^2}{n^2} \enspace .
\end{align*}
Let \C{i} be a color and \notC{i} be the set of all other colors, defined as
$\notC{i} = C \setminus \C{i}$. We observe that after one round the new number
of nodes supporting \C{i} is a random variable
\begin{equation*}
c_i' = c_i + \sum_{j \neq i} f_{ji} - \sum_{j \neq i} f_{ij} = c_i + f_{\notC{i}i} - f_{i\notC{i}} \enspace.
\end{equation*}

Since all nodes perform their choices independently, the first sum $f_{\notC{i}i}$ has a
binomial distribution with parameters $f_{\notC{i}i} \sim
B(n-c_i,\,\ifrac{c_i^2}{n^2})$. Furthermore, every node of color \C{i} changes
its color away from \C{i} to any other opinion with probability
$p^{\text{away}}_i = \sum_{j\neq i}\ifrac{c_j^2}{n^2}$. Therefore, the second
sum $f_{i\notC{i}}$ also has a binomial distribution with parameters
$f_{i\notC{i}} \sim B(c_i,\,p^{\text{away}}_i)$. That is, we have in
expectation
\begin{equation} \label{eq:expected-ci}
\Expected{c_i'} =c_i + \frac{\left(n-c_i\right)c_i^2}{n^2} - \frac{c_i}{n^2}\sum_{j\neq i}c_j^2 \enspace .
\end{equation}
Note that these expected values are monotone w.r.t.\ the current size. This is
described more formally in the following observation.
\begin{observation} \label{obs:monotone}
Let $\C{r}$ and $\C{s}$ be two colors. It holds that if $c_r \leq c_s $ then $ \Expected{c_r'} \leq \Expected{c_s'}$.
\end{observation}
\begin{proof}
We first rewrite \eqref{eq:expected-ci} as
\begin{equation*}
\Expected{c_i'} =c_i + \frac{c_i^2}{n} - \frac{c_i}{n^2}\sum_{\C{j}}c_j^2 = c_i \left(1 + \frac{c_i}{n} - \sum_{\C{j}}\frac{c_j^2}{n^2} \right)\enspace .
\end{equation*}
Using this representation of $\Expected{c_i'}$ gives us
\begin{equation*}
\Expected{c_r'} =  c_r \left(1 + \frac{c_r}{n} - \sum_{\C{j}}\frac{c_j^2}{n^2} \right) \stackrel{(c_r \leq c_s)}{\leq} c_s \left(1 + \frac{c_s}{n} - \sum_{\C{j}}\frac{c_j^2}{n^2}\right) = \Expected{c_s'} \enspace . \qedhere
\end{equation*}
\end{proof}

For the following lemma, recall that $\A = \C1$ denotes the dominant color of
size $a=c_1$ and $\B = \C2$ denotes the second largest color of size $b = c_2$.

\begin{lemma} \label{lem:distance-increases}
Let \A be the dominating color and \B be the second largest color. Assume that
$a - b > z \cdot \sqrt{n\log{n}}$. There exists a constant $z$ such that
$a'-b'>\left(a-b\right)\left(1+\ifrac{a}{4n}\right)$ \whp.
\end{lemma}

In the following proof we utilize certain methods which have also been used in
\cite{CER14} for the two-opinion plurality consensus process with two choices
in more general graphs.

\begin{proof}
First we observe that
{
\allowdisplaybreaks
\begin{align*}
\Expected{a'-b'} &= a+\Expected{f_{\overline{\A}\A}} - \Expected{f_{\A\overline{\A}}} - b - \Expected{f_{\overline{\B}\B}} + \Expected{f_{\B\overline{\B}}}\\
&= a + \left(n-a\right)\cdot\frac{a^2}{n^2} - \frac{a}{n^2}\sum_{\C{i} \neq \A} c_i^2 -b - \left(n-b\right)\cdot\frac{b^2}{n^2} + \frac{b}{n^2}\sum_{\C{i} \neq \B} c_i^2 \\
&= a-b + \frac{1}{n^2}\left(a^2n-a^3-b^2n+b^3-a\sum_{\C{i} \neq \A} c_i^2 + b\sum_{\C{i} \neq \B} c_i^2 \right)\\
&= a-b + \frac{1}{n^2}\left(n\left(a^2-b^2\right)-a\left(a^2+\sum_{\C{i} \neq \A}c_i^2\right)+b\left(b^2+\sum_{\C{i} \neq \B}c_i^2\right)\right) \\
&= a-b + \frac{1}{n}\left(a^2-b^2\right) - \frac{1}{n^2}\left(a\sum_{\C{i}}c_i^2 - b\sum_{\C{i}}c_i^2\right) \\
&= a-b + \frac{\left(a-b\right)\left(a+b\right)}{n} - \frac{1}{n^2} \sum_{\C{i}}c_i^2 \left(a -  b\right) \\
&= \left(a-b\right)\cdot\left(1 + \frac{\left(a+b\right)}{n} - \frac{1}{n^2}\sum_{\C{i}}c_i^2 \right) \enspace .\\
\intertext{We now use that \A and \B are the largest and second largest colors, respectively, to bound the sum $\sum_{\C{i}}c_i^2$ as follows.
\begin{equation*}
\sum_{\Ci}c_i^2 = a^2 + \sum_{\Ci\neq\A}c_i^2 \leq a^2 + \sum_{\Ci\neq\A}c_i\cdot b = a^2 + \left(n-a\right)\cdot b \leq a^2 + n \cdot b
\end{equation*}
Therefore, we obtain}
\Expected{a'-b'} & \geq \left(a-b\right)\left(1 + \frac{ \left(a+b\right) }{n} - \frac{a^2+n\cdot b}{n^2} \right) \\
&\geq \left(a-b\right)\left(1+\frac{a}{n}\cdot\left(1-\frac{a}{n}\right)\right) \\
\intertext{and since $a \leq \ifrac{n}{2}$ we finally get}
\Expected{a'-b'} & \geq \left(a-b\right)\left(1+\frac{a}{2n}\right) \enspace .
\end{align*}
}

\noindent We now apply Chernoff bounds to $a'-b'$. Let $\delta_1$, $\delta_2$,
$\delta_3$, $\delta_4$ be defined as
\begin{align*}
\delta_1 &= \frac{2\sqrt{n\log{n}}}{a} \enspace , &
\delta_2 &= \frac{2n\sqrt{\log{n}}}{\sqrt{a\sum_{\C{i}\neq\A}c_i^2}} \enspace , &
\delta_3 &= \frac{2\sqrt{n\log{n}}}{b} \enspace , &
\delta_4 &= \frac{2n\sqrt{\log{n}}}{\sqrt{b\sum_{\C{i}\neq\B}c_i^2}}
\intertext{for the corresponding random variables $f_{\overline{\A}\A}$, $f_{\A\overline{\A}}$, $f_{\overline{\B}\B}$, $f_{\B\overline{\B}}$ with expected values $\mu_1, \mu_2, \mu_3, \mu_4$ given by}
\mu_1 &= \left(n-a\right)\frac{a^2}{n^2} \enspace , &
\mu_2 &= \frac{a}{n^2}\sum_{\C{i} \neq \A}c_i^2 \enspace , &
\mu_3 &= \left(n-b\right)\frac{b^2}{n^2} \enspace , &
\mu_4 &= \frac{b}{n^2}\sum_{\C{i} \neq \B}c_i^2 \enspace .
\end{align*}
Since $a \leq \ifrac{n}{2}$ we know for the second largest color \B that $b
\geq \ifrac{n}{2k}$. Together with $a\geq \ifrac{n}{k} \geq n^{1-\varepsilon}$ we get $0 < \delta_i < 1$ and
$\delta_i^2\cdot\mu_i=\BigOmega{\log{n}}$ for $i = 1,2,3,4$. We now apply
Chernoff bounds to $a'-b'$ and obtain \whp
\begin{equation*}
a'-b' \geq \left(a-b\right)\cdot\left(1+\frac{a}{2n}\right) - E
\end{equation*}
where the error term $E$ is bounded as follows.
\begin{align*}
E & = \delta_1\cdot\mu_1 + \delta_2\cdot\mu_2 +\delta_3\cdot\mu_3 +\delta_4\cdot\mu_4 \\
& = \frac{2\sqrt{n\log{n}}}{n^2}\left(an-a^2 + \sqrt{an\sum_{\C{i} \neq \A}c_i^2} +bn-b^2 +\sqrt{bn\sum_{\C{i} \neq \B}c_i^2} \right)\\
& \leq \frac{2\sqrt{n\log{n}}}{n^2}\left(\sqrt{n\sum_{\C{i}}c_i^2}\left(\sqrt{a} + \sqrt{b}\right) +an +bn \right)\\
& \leq \frac{2\sqrt{n\log{n}}}{n^2}\left(2an +an +bn \right)\\
& \leq \frac{8 a\sqrt{n\log{n}}}{n} \enspace ,
\end{align*}
where we used that $\sum_{\C{i}}c_i^2\leq \sum_{\C{i}}a\cdot c_i\leq an$.
From the conditions in the statement of the lemma we know that $\left(a-b\right) \geq z \cdot
\sqrt{n\log{n}}$ for some constant $z$. If we assume that $z$ is large enough,
e.g., $z \geq 32$, then we get \whp
\begin{equation*}
a'-b' \geq \left(a-b\right)\cdot\left(1+\frac{a}{4n}\right) \enspace . \qedhere
\end{equation*}
\end{proof}

While \autoref{lem:distance-increases} shows that in the absence of an adversary the difference
between colors \A and \B does indeed increase in every round \whp, it does not cover
the remaining colors \C{j} for $j \geq 3$, nor does it address the presence of an adversary. To show that the smaller colors
\C{j} do also not interfere with \A and thus the minimum of the difference between \A
and any \C{j} increases, we use the following coupling. 

At any time step $t$, there exists a bijective function which maps any instance
of the two-choices protocol at time $t$ to another instance of the same
protocol such that the outcome $c'$ of the first instance is at most the
outcome $b'$ of the mapped instance.

\begin{lemma} \label{lem:coupling}
Let \A be the dominating color of size $a$ and let \B be the second largest
color of size $b$. Let $\C{} \neq \A,\B$ be one of the remaining colors of size
$c$.
Furthermore, let $\pi : V \rightarrow V$ be a bijection and let $P$
be the original process. We can couple a process  $P' = P(\pi)$ to the original
process $P$ such that $\P{c'}\leq \PP{b'}$, where $\P{c'}$ is the random variable
$c'$ in the original process and $\PP{b'}$ is the random variable $b'$ in
the coupled process.
\end{lemma}

\begin{proof}
{
\def\C{\ensuremath{\mathcal{C}}\xspace}
Let $t$ be an arbitrary but fixed round. In the following, we use the notation
that $\B_t$ and $\C_t$ are sets containing all vertices of colors \B and \C{},
respectively, in round $t$. As before, we have color sizes $b = |\B_t|$ and $c
= |\C_t|$. The proof proceeds by a simple coupling argument. We start by
defining $\hat \B_t, \B_t^*, \C_t^* \subseteq V$ as follows. Let $\hat\B_t$ be
an arbitrary subset of $\B_t$ such that $|\hat\B_t|=|\C_t|$. Let furthermore
$\B_t^*$ be defined as $\B_t^* = \B_t \setminus \hat \B_t$, and finally let
$\C_t^*$ be an arbitrary subset of $V\setminus \left(\B_t \cup \C_t\right)$
such that $|\C_t^*|=|\B_t^*|$.

Additionally, we construct the bijective function $\pi : V \rightarrow V$ as
follows. Let $\hat\pi$ be an arbitrary bijection between $\C_t$ and $\hat\B_t$.
Let furthermore $\pi^*$ be an arbitrary bijection between $\C_t^*$ and
$\B_t^*$. We now define $\pi$ as

\begin{equation}
\pi(v) = \begin{cases}
\hat\pi(v) & \text{ if }v \in \C \enspace , \\
\hat\pi^{-1}(v) & \text{ if }v \in \hat\B \enspace , \\[1ex]
\pi^*(v) & \text{ if }v \in \C^* \enspace ,\\ 
{\pi^*}^{-1}(v) & \text{ if }v \in \B^* \enspace ,\\[1ex]
v & \text{ if }v \in V \setminus (\B_t \cup \C_t \cup \C_t^*)\enspace .
\end{cases}
\label{eq:bijective-pi}
\end{equation}
A graphical representation of $\pi$ can be seen in \autoref{fig:bijective-pi}.

\begin{figure}[h]
\centering
\includegraphics{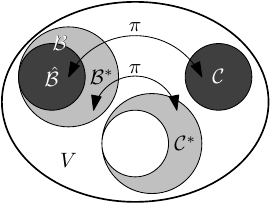}
\caption[schematic representation of the coupling]{schematic representation of the bijective function $\pi$
defined in \eqref{eq:bijective-pi}}
\label{fig:bijective-pi}
\end{figure}

It can easily be observed that $\pi$ indeed forms a bijection on $V$.  We now
use $\pi$ to couple a process $P' = P(\pi)$ to the original process $P$, to
show that $\PP{b'} \geq \P{c'}$, where the notation $\P{b'}$ means the variable
$b'$ in the original process $P$ and $\PP{c'}$ means the variable in the
coupled process $P'$. Let $u \in V$ be an arbitrary but fixed node. The
coupling is now constructed such that whenever $u$ samples a node $v \in V$ in
the original process $P$, then $u$ samples $\pi(v)$ in the coupled process
$P'$. 

Let $X$ be the set of nodes changing their opinion to \C from any other
color in $P$, that is,
\begin{equation*}
X = \left\{v \in V: v \notin \C_t \wedge v \in \C_{t+1}\right\} \enspace .
\end{equation*}
Clearly, $X$ consists of two disjoint subsets $X = \hat{X} \cup X^*$, defined as
\begin{align*}
\hat{X} &= \left\{v \in V: v \notin \left(\C_t \cup \C^*_t\right) \wedge v \in \C_{t+1}\right\} \\
\lefttag{and}
X^* &= \left\{v \in V: v \in \C^*_t \wedge v \in \C_{t+1}\right\} \enspace .
\end{align*}
The set $\hat{X}$ consists of all nodes changing their opinion to \C from
any other color except $\C^*$. The set $X^*$ contains the
remaining nodes in $\C^*$ changing their opinion to \C. Analogously to $X$, let $Y$ be the set of nodes changing their opinion from
\C to any other color in $P$, that is,
\begin{equation*}
Y = \left\{v \in V: v \in \C_t \wedge v \notin \C_{t+1}\right\} \enspace .
\end{equation*}
Again, we have $Y = \hat{Y} \cup Y^*$ which are defined as
\begin{align*}
\hat{Y} &= \left\{v \in V: v \in \C_t \wedge v \notin \left(\C_{t+1}\cup\C^*_{t+1}\right)\right\} \\
\lefttag{and}
Y^* &= \left\{v \in V: v \in \C_t \wedge v \in \C^*_{t+1}\right\} \enspace .
\end{align*}

\begin{table}
\smaller
\noindent\begin{tabularx}{\textwidth}{lXX}
Set       & Process $P$                                                     & Process $P'$ \\\hline
$X$       & nodes changing their color to $\C$                          & nodes which now belong to $\hat\B$ \\
$\hat{X}$ & nodes changing their color to $\C$ except nodes from $\C^*$ & nodes changing their color to \B \\
$X^*$     & nodes from $\C^*$ changing their color to \C                & nodes changing their color to \B \\
$Y$       & nodes changing their color from $\C$                   & nodes which no longer belong to $\hat\B$ \\
$\hat{Y}$ & nodes changing their color from $\C$ but not to $\C^*$ & nodes changing their color  from $\hat{\B}$ but not to $\B^*$ \\
$Y^*$     & nodes changing their color from $\C$ to $\C^*$              & nodes changing from $\hat{\B}$ to $\B^*$ \\
\end{tabularx}
\caption[corresponding sets in the coupling]{corresponding sets between processes $P$ and $P'$}
\label{tab:correspondences}
\end{table}

We now analyze the behavior of these sets in the coupled process $P'$. The coupling ensures the correspondences described in \autoref{tab:correspondences}.
We therefore have in $P$
\begin{align}
\P{c'~} &= \P{c} + |X| - |Y| \enspace . \label{eq:coupling-P} \\
\intertext{In $P'$, we first observe that $|\B| = |\hat{\B}| + |\B^*|$ and therefore}
\PP{b'~} &\geq \PP{b} + |\hat{X}| - |\hat{Y}| - \left(|\B^*| - |X^*|\right) \label{eq:coupling-PP-bound} \\
&\geq |\hat\B| + |\B^*| + |\hat{X}| - |\hat{Y}| - |\B^*| + |X^*| \notag\\
&= |\hat\B| + |X| - |\hat{Y}| \notag\\
&\geq |\hat\B| + |X| - |Y| \notag\\
&= \P{c} + |X| - |Y| \label{eq:coupling-PP}
\end{align}
where the expression $|\B^*|-|X^*|$ in \eqref{eq:coupling-PP-bound} is an upper bound on the
number of nodes in $\B^*$ changing their color away from $\B$ to any other
color except $\hat{\B}$. Combining equations \eqref{eq:coupling-P} and
\eqref{eq:coupling-PP} gives us
\begin{equation*}
\P{c'} \leq \PP{b'}
\end{equation*}
which concludes the proof.
}
\end{proof}

We now use \autoref{lem:distance-increases} and \autoref{lem:coupling} to prove
\autoref{thm:main-result}.

\begin{proof}
Let $\A = \C{1}$ be the dominant color and $\B = \C{2}$ the second largest
color. Assume $a - b \geq z \cdot \sqrt{n\log{n}}$ for a sufficiently large
constant $z$. From \autoref{lem:distance-increases} we know that $a' - b' \geq
\left(a-b\right)\cdot\left(1+\ifrac{a}{4n}\right)$ \whp. 
Since \B is the second largest color, we obtain from \autoref{lem:coupling} for any remaining color \C{j} with $j \geq 3$ that \whp
$ a' - c_j' \geq a' - b' \geq
\left(a-b\right)\cdot\left(1+\ifrac{a}{4n}\right)$. Note that it may very well
happen, especially if all colors have the same size except for \A, that another
color \C{j} \emph{overtakes} \B. However, the resulting distance between
\A and this new second largest color \C{j} will be larger than
$\left(a-b\right)\cdot\left(1+\ifrac{a}{4n}\right)$ \whp.
Let $a''$ and $b''$ denote the sizes of the colors after the round, that is, after the adversary changed the opinion of up to $F$ arbitrary nodes.
We have $a''-b'' \geq a'-b' - 2F \geq \left(a-b\right)\cdot\left(1+\ifrac{a}{4n} - \ifrac{2F}{a-b}\right)\geq \left(a-b\right)\cdot\left(1+\ifrac{a}{8n}\right)$, since $F=\ifrac{a(a-b)}{8n}$.

Taking the union bound over all colors, we conclude that the distance between the first color \A and every other color grows in every round by a factor of at least
$\left(1+\ifrac{a}{4n}\right)$ \whp. Therefore, after
$\tau= \ifrac{4n}{a}$ rounds, the relative distance between \A and \B doubles \whp.
Hence, the required time for \A to reach a size of at least
$(\ifrac{1}{2}+\eps{1})\cdot n$ for a constant $\eps{1} > 0$ is
bounded by $\BigO{\ifrac{n}{a}\cdot\log{n}}$.
This bias is large enough that we assume in the following  that all nodes which are not of color $\A$ are of color $\B$.
In absence of an adversary, we can see that after additional $\BigO{\log{n}}$
rounds every node has the same color \A, \whp; see \cite{CER14}.
 In each individual round, the growth
described in \autoref{lem:distance-increases} takes place \whp. A union bound
over all \BigO{\ifrac{n}{a} \cdot \log{n}} rounds yields that the protocol indeed converges to \A within \BigO{\ifrac{n}{a} \cdot \log{n}} rounds
\whp.
The same analysis of \cite{CER14} can be used even in the presence of an adversary. However, in this case
we can only reach \emph{almost validity} according to \autoref{def:stability}, since the adversary is allowed to change $F=\LittleO{n}$ nodes per round.

Finally, we argue that the two-choices process trivially fulfills the property \emph{almost agreement} according to \autoref{def:stability}.
Starting from an arbitrary initial distribution of colors, there is in every round a positive (albeit super-exponentially small in $n$) probability that all nodes adopt the same color.
\end{proof}

\subsection{Lower Bounds}\label{sect:lowerbound}

In the previous section, we showed that the plurality consensus process \whp
converges to \A if the initial imbalance $a - b$ is not too small. Precisely,
\autoref{thm:main-result} states that if $a-b \geq z \cdot \sqrt{n\log{n}}$ for
some constant $z$, \A wins \whp. Conversely, in the following section we
examine a lower bound on the initial bias. We will show, as stated in
\autoref{thm:lowerbound-bias}, that for an initial bias $a-b \leq z \cdot
\sqrt{n}$  for some constant $z$ we have a constant probability that \B
\emph{overtakes} \A in the first round, that is, $\Probability{a' < b'} =
\BigOmega{1}$.

Our proof of \autoref{thm:lowerbound-bias} is based on the normal approximation
of the binomial distribution. In this context, we adapt Theorem 2 and equation
(6.7) from \cite{Fel68} as stated in the following theorem.

\begin{theorem}[DeMoivre-Laplace limit theorem \cite{Fel68}] \label{thm:deMoivre}
Let $X$ be a random variable with binomial distribution $X \sim B(N,p)$. It
holds for any $x > 0$ with $x = \LittleO{N^{\ifrac{1}{6}}}$ that
\begin{equation*}
\Probability{ X \geq \Expected{X} + x\cdot\sqrt{\Variance{X}}}= \frac{1}{\sqrt{2\pi}\cdot x}\cdot \Exp{-{x^2}{2}} \pm \LittleO{1} \enspace .
\end{equation*}
\end{theorem}

We now use \autoref{thm:deMoivre} and prove \autoref{thm:lowerbound-bias} which states that
there exists an initial color assignment for which
$a=b+z'\cdot\sqrt{n}$ but color \B wins with constant probability  even in absence of an adversary. 
\begin{theorem}[Lower Bound on the Initial Bias] \label{thm:lowerbound-bias}
For any $k\leq \sqrt{n}$ and constant $z'$ there exists an initial assignment of colors to nodes for which
$a=b+z'\cdot\sqrt{n}$ but $\Probability{a' < b'} = \BigOmega{1}$ even in absence of an adversary.
\end{theorem}
\begin{proof}
{
\def\fab{\ensuremath{f_{\A\B}}\xspace}
\def\fba{\ensuremath{f_{\B\A}}\xspace}
\def\fbb{\ensuremath{f_{\B\overline\B}}\xspace}
\def\a{\left(n' + z\cdot\sqrt{n}\right)}
\def\b{\left(n' - z\cdot\sqrt{n}\right)}
Let $z=z'/2$ and $n'=\frac{n-k+2}{2}$.
Assume that we have the following initial color distribution among the nodes.
\begin{equation*}
(c_1, c_2, c_3, \dots, c_k) = \left( \floor{n'} + \floor{z\cdot\sqrt{n}} , \ceil{n'} - \floor{z\cdot\sqrt{n}} , 1, \dots, 1\right).
\end{equation*}
Clearly, $\sum_{\C{j}}c_j = n$. In the following we will omit the floor and
ceiling functions for the sake of readability reasons. First, we start by
giving an upper bound on the number of nodes which change their color away from
\B.
Now recall that \fbb follows a binomial distribution 
$\fbb \sim B(b,\sum_{C{j} \neq B}\ifrac{c_j^2}{n^2})$ with expected value
\begin{align*}
\Expected{\fbb} &= b\cdot\frac{a^2+k-2}{n^2} \\
&= \b\cdot\frac{\a^2 + k-2}{n^2}\\
&\leq \frac{\a^3 + k-2}{n^2}\\
&\leq \frac{n}{8} + 4z\sqrt{n}%
 \enspace .
\end{align*}
Applying Chernoff bounds to \fbb gives us
\begin{equation}
\Probability{\fbb \geq \left(1 + \sqrt{3/\Expected{\fbb}}\right)\cdot\Expected{\fbb}} \leq \ifrac{1}{e} \enspace . \label{eq:lowerbound-bias-bb}
\end{equation}
That is, with constant probability at least $1 - \ifrac{1}{e}$ we have
\begin{align*}
\fbb &\leq \left(1+\sqrt{3/\Expected{\fbb}}\right)\cdot\Expected{\fbb} \\
& \leq \frac{n}{8} + 4z\sqrt{n}+ \sqrt{3\cdot\Expected{\fbb}} \\
& \leq \frac{n}{8} + \left(4z+1\right)\cdot\sqrt{n} \enspace .
\end{align*}

Secondly, we give the following lower bound on the number of nodes which change
their color from \A to \B. Again, the random variable \fab denoting the flow from \A
to \B has a binomial distribution $\fab \sim B\left(a, \ifrac{b^2}{n^2}\right)$
with expected value
\begin{align*}
\Expected{\fab} &= \a\cdot\frac{\b^2}{n^2} \\
& \geq \frac{\b^3}{n^2} \\
& \geq \frac{(n/2-(z+1/2)\sqrt{n})^3}{n^2} \\
& \geq \frac{n}{8} - 4z\sqrt{n}%
\intertext{and variance}
\Variance{\fab} &= \Expected{\fab}\cdot\left(1-\frac{\b^2}{n^2}\right) \\
& \geq \frac{n}{9}\cdot \frac{1}{2} = \frac{n}{18} \enspace .
\intertext{We now apply \autoref{thm:deMoivre} to $\fab$. Let $x=\frac{\sqrt{18}}{2}(18z+4)$. We derive
\begin{equation*}
\Probability{\fab \geq \Expected{\fab} + x\cdot\sqrt{\Variance{\fab}}} = \frac{1}{\sqrt{2\pi}\cdot x} \Exp{-{x^2}/{2}} \pm \LittleO{1} = \BigOmega{1} \enspace .
\end{equation*}
That is, we have with constant probability}
\fab &\geq \Expected{\fab} + x\cdot\sqrt{\Variance{\fab}} 
\geq  \frac{n}{8} - 4z\sqrt{n} + x\cdot\sqrt{\frac{n}{18}} 
\numberthis \label{eq:lowerbound-bias-ab} \enspace .\\
\end{align*}

Finally, assume that in the worst case every node of colors $\C3, \dots, \C{k}$
changes to \A but not a single node changes away from \A to these colors \C3 to
\C{k}. Observe that \fbb is an upper bound on \fba. Therefore, 
\begin{align*}
a'-b' &\leq \left( a + (k-2) + \fba -\fba \right) -\left( b+ \fab  - \fbb  \right) \\
&\leq a-b + (k-2) + 2\fbb-2\fab \\
&\leq 2z\cdot\sqrt{n} + (k-2) + 2\fbb-2\fab \\
&\leq (2z+1)\cdot\sqrt{n}+ 2\fbb-2\fab \enspace .
\intertext{We plug in \eqref{eq:lowerbound-bias-bb} and \eqref{eq:lowerbound-bias-ab} to
bound the random variables \fab and \fbb and obtain with constant probability}
a'-b' &\leq (2z+1)\cdot\sqrt{n} + 2\left(  \frac{n}{8} + (4z+1)\sqrt{n} \right) - 2\left( \frac{n}{8} - 4z\sqrt{n} + x\cdot\sqrt{\frac{n}{18}} \right)\\
&= (2z +1 + 8z + 2 + 8z - 2x/\sqrt{18})\cdot\sqrt{n}  \\
&= (18z +3 - 2x/\sqrt{18})\cdot\sqrt{n}
\end{align*}
which gives us $a'-b'< 0$ for $x = \frac{\sqrt{18}}{2}(18z+4)$.
Therefore, we have $\Probability{a' < b'} = \BigOmega{1}$ and thus we conclude that color \B wins with constant probability. \qedhere
}
\end{proof}

\begin{theorem}[Lower Bound on the Run Time] \label{thm:lowerbound-runtime}
Assume the initial bias is exactly $z\sqrt{n\log n}$ for some constant $z$. The number
of rounds required for the plurality consensus process defined in
\autoref{alg:two-choices} to converge is at least
$\BigOmega{\ifrac{n}{a} + \log n}$ with constant probability, even in absence of an adversary.
\end{theorem}

\begin{proof}
Let $a(t)$ denote the size of color \A in round $t$. Assume \A is the largest
color of initial size $a(0) = {n}/{k}+z\cdot\sqrt{n\log{n}}$. Furthermore,
assume that $k\geq 3\cdot z$.  We show by induction on the rounds that
$a(t) \leq a(0)\cdot\left(1+ 3\cdot{a(0)}/{n}\right)^t$ for
$1 \leq t \leq {n}/{(10\cdot a(0))}$ with probability $1-t/n$.
First we note that 
\begin{align*} \label{eq:coarse-bound}
a(t) &\leq a(0) \cdot \left(1+ 3\cdot\frac{a(0)}{n}\right)^t\\
&\leq a(0)\cdot \left(1+ 3\cdot\frac{a(0)}{n}\right)^{n/(10 \cdot a(0))}\\
&\leq a(0)\cdot \Exp{1/2} \\
&\leq 2\cdot a(0) \numberthis
\intertext{and}
a(t)&\geq a(0)\enspace . \label{eq:coarse-bound-2}\numberthis
\end{align*}
We now prove the induction claim. The base case holds trivially.
Consider step ${t+1}$.
By induction hypothesis we have with probability at least $1-{t}/{n}$ that 
$ a(t) \leq a(0) \cdot\left(1+ 3\cdot{a(0)}/{n}\right)^t$.
Note that we have \whp \begin{align*}
a(t+1) &\leq a(t) + f_{\overline\A\A} \\
&\leq a(t) + \left(1+\frac{\sqrt{3\log{n}}}{\sqrt{\Expected{f_{\overline\A\A}}}}\right)\cdot\Expected{f_{\overline\A\A}} \enspace ,
\intertext{where the latter inequality follows by Chernoff bounds. Using \eqref{eq:coarse-bound} and \eqref{eq:coarse-bound-2}, we derive}
a(t+1) & \leq a(t) + \left(1+\frac{\sqrt{3\log n}}{\sqrt{ {a(t)^2}/{(2\cdot n)} }}\right)\frac{a(t)^2}{n} \\
& \leq a(t) + \left(1+\frac{\sqrt{3\log n}}{\sqrt{{a(0)^2}/{(2\cdot n)} }}\right)\frac{a(t)^2}{n} \\
& \leq a(t) + \frac{3}{2} \cdot \frac{a(t)^2}{n}  \\
& = a(t) \cdot \left(1 + \frac{3}{2}\cdot \frac{a(t)}{n}\right) \\
& \leq a(t) \cdot \left(1 + \frac{3\cdot a(0)}{n}\right) \enspace .
\intertext{From the induction hypothesis we therefore obtain}
a(t+1)& \leq a(0)\cdot\left(1+ \frac{3\cdot a(0)}{n}\right)^t\cdot \left(1 + \frac{3\cdot a(0)}{n}\right) \\
& = a(0)\cdot\left(1+ \frac{3\cdot a(0)}{n}\right)^{t+1} \enspace .
\end{align*}
Using a union bound to account for all errors, we derive that with probability at
least $1-{(t+1)}/{n}$ we have 
$a(t+1) \leq  a(0)\cdot\left(1+ {3\cdot a(0)}/{n}\right)^{t+1}$,
which completes the proof of the induction and proves the lower bound of $\BigOmega{{n}/{a}}$.

In the remainder we establish the bound $\BigOmega{\log n}$.
Assume only two colors \A and \B, where \A is the largest
color of initial size $a(0) = {n}/{2}+\sqrt{n}\log{n}$.  We show by induction on the rounds that
$a(t) \leq a(0) + 6^t\sqrt{n}\log{n}$ for
$1 \leq t \leq {\log n}/{20}$ with probability $1-{2t}/{n}$.
First we note that 
\begin{align*}
a(t) &\leq a(0)+ 6^t\sqrt{n}\log{n} \leq {n}/{2}+n^{5/6}<n
\intertext{and }
a(t)&\geq a(0)\enspace .
\end{align*}
We now prove the induction claim. The base case holds trivially.
Consider step ${t+1}$.
By induction hypothesis we have with probability at least $1-{2t}/{n}$ that 
$ a(t) \leq a(0) + 6^t\sqrt{n}\log{n}$.
We have, using $a=a(t)$ and $\beta = 6^t\sqrt{n}\log{n}$,
\begin{align*}
	n^2\cdot \Expected{f_{\overline\A\A}-f_{\A\overline\A}}&=(n-1)a^2-a\cdot (n-a)^2=(n-a)a(2a-n)\\
	&\leq{n}/{2}\cdot a \cdot 2\beta =n\cdot \beta (n+\beta)=n^2\cdot \beta + n \cdot \beta^2 \enspace .
\end{align*}
Similar to before, we obtain by Chernoff bounds that \whp
\begin{align*}
a(t+1)-a(t)&=	f_{\overline\A\A}-f_{\A\overline\A} \\
& \leq \left(1+\frac{\sqrt{3\log{n}}}{\sqrt{\Expected{f_{\overline\A\A}}}}\right)\Expected{f_{\overline\A\A}} 
- \left(1-\frac{\sqrt{3\log{n}}}{\sqrt{\Expected{f_{\A\overline\A}}}}\right)\Expected{f_{\A\overline\A}}\\
&\leq  \Expected{f_{\overline\A\A}-f_{\A\overline\A}}+   2\sqrt{3\log{n}}\cdot \sqrt{\Expected{f_{\overline\A\A}}} \\
&\leq \beta +  \beta^2/n +   2\sqrt{3\log{n}}\cdot \sqrt{n} \leq 3\beta \enspace .
\intertext{From the induction hypothesis we therefore obtain}
a(t+1)& \leq a(0) +6^t\sqrt{n}\log{n} + 3\beta  \\ &
\leq   a(0) +6^{t+1}\sqrt{n}\log{n}  \enspace ,
\end{align*}
which completes the induction and yields the lower bound of $\BigOmega{\log n}$.
\end{proof}

\section{Comparison with the 3-Majority Process }\label{sec:3majority}

In this section we elaborate on the difference between the two-choices process
and the $3$-majority rule \cite{BCNPST14}, where in the latter each node pulls
the opinion of three random neighbors and adopts the majority opinion among
those three, breaking ties uniformly at random. As mentioned before, the
$3$-majority process of \cite{BCNPST14} uses $\BigO{\log k}$ memory bits and
the authors prove a tight run time of $\BigTheta{k\cdot\log n}$ for this
protocol, given a sufficiently high bias $c_1-c_2$. Moreover, they show that if
the bias is only of order $\sqrt{kn}$, then with constant probability the
difference $c_1-c_2$ decreases. This is fundamentally different to the
two-choices process, where we only require a bias of $\BigOmega{\sqrt{n
\log{n}}}$.

The reasons are the following. First, the variance in the $3$-majority process
can be orders of magnitude larger and second, the expected increase in the
difference between the largest and second largest color in the $3$-majority
process is only of order of the variance. As for the variance, consider an
initial setting where all colors are of sublinear size and $\A$ and $\B$ are
larger than all other colors, such that 
\begin{align*}
\LittleO{n} = a &= b + c\sqrt{n\log n} > c_j + c\sqrt{n\log n}
\intertext{ and }  c_j&=(n-b-a)/(k-2)
\end{align*}
for all $2\leq j\leq k$ with $k=n^\varepsilon$ for constants $\varepsilon$ and
$c$. Observe that the expected numbers of color switches differ significantly. In the
two-choices process it is very unlikely for a node to pick
the same color twice and the probability of switching is $\LittleO{1}$. In contrast to this,
the probability of switching in the $3$-majority process is $1-\LittleO{1}$.

More illustratively, consider the number of switches to color $\B$. By Lemma 2.1
of \cite{BCNPST14}, the probability that a node switches to color $\B$ in the $3$-majority
process is $p \in [b/n,2b/n]$ and the variance becomes $n\cdot
p\cdot (1-p) \geq b/2$. However, in the two-choices process, the probability of
switching to $\B$ is $q= b^2/n^2$ and the variance is thus at most $n\cdot
q\cdot (1-q) \leq n\cdot q = b^2/n$, which is considerably smaller than $b/2$.
This high variance paired with the small expected increase in the difference
between $\A$ and $\B$ easily becomes fatal. Again, by Lemma 2.1 of
\cite{BCNPST14}, one can verify that $\Expected{a'-b'}\leq
a-b+{(a^2-b^2)}/{n}$. Now we have $\Probability{a'\leq
\Expected{a'}}=\BigOmega{1}$ and, using the large variance, we obtain from
\autoref{thm:deMoivre} that
\begin{align*}
\Probability{b' \geq b + (a^2-b^2)/n \left| a'\leq \Expected{a'}\right.} &\geq \Probability{b' \geq b + (a^2-b^2)/n}\\
&\geq \Probability{b' \geq \Expected{b'} + \Variance{b'}} =\BigOmega{1} \enspace .
\end{align*}
Thus the distance between \A and \B \emph{decreases} with constant probability, that
is, $\Probability{ a'-b' < a-b }=\BigOmega{1}$. In comparison to this, we have
seen in \autoref{sect:main-result} that in the given setting the
distance between \A and \B in the two-choices process increases \whp.

\section{Analysis of the Synchronous Algorithm: One Extra Bit} \label{sect:memory}

\newcommand\cycle{phase\xspace}
In this section we investigate the \onebit protocol which combines the
guarantees of the two-choices process to reach plurality consensus with the
speed of broadcasting.  The protocol consists of $\BigTheta{\log(n/a) +
\log\log{n}}$ phases which in turn consist of two sub-phases, one round of the
\TC process and multiple rounds of  the so-called \BP sub-phase. In the latter
\BP sub-phase, each node that changed its opinion during the preceding
two-choice round broadcasts its new opinion.

More precisely, we consider the modified model where each node is allowed to
store and transmit one additional bit. This bit is set to \ttrue if and only if
a node changed its opinion in the \TC sub-phase. In the \BP sub-phase, each
node $u$ samples nodes randomly until  a node  $v$ with a bit set to \ttrue is
found. Then $u$ adopts $v$'s opinion and sets its own bit to \ttrue, which
means that subsequently any node sampling $u$ will set their bit directly.

The first sub-phase ensures that in a round $t$ the number of nodes holding
opinion \A and having their bit set to \ttrue is concentrated around
$a_{t-1}^2/n$. After the \BP sub-phase, all nodes will have their bit set, and
the distribution and the size of \A's support is concentrated around
$a^2/x(1)$, where $x(1)$ is the total number of bits set after the \TC
sub-phase. Moreover, we show that no other color grows faster. In fact, we show
that the distance between $\A$  and any opinion $\C{j}\neq A$ increases
quadratically, that is, $a'/c_j' \geq (1-\LittleO{1})\cdot a^2/c_j^2$. Due to
the  quadratic growth in the distance between $\A$ and every other opinion, the
number of phases required is only of order $\BigTheta{\log(n/a) +
\log\log{n}}$. The process runs in multiple \cycle{s} of length
$\BigTheta{\log{k} + \log\log{n}}$ each, therefore we assume that every node is
aware of (upper bounds on) $n$ and $k$. The process is formally defined in
\autoref{alg:memory}.

\providecommand\cycle{phase\xspace}
\begin{algorithm}[t]
\SetAlgoVlined
\SetKwFor{AtNode}{at each node}{do in parallel}{end}
\SetKw{KwLet}{let}

\SetKwProg{Algorithm}{Algorithm}{}{end}
\SetKwFunction{Memory}{memory}
\SetKwFunction{Color}{color}
\SetKwFunction{Bit}{bit}

\SetKw{KwLet}{let}
\Algorithm{\Memory{$G = (V, E)$, $ \Color : V \rightarrow C$, $\Bit : V \rightarrow \{\ttrue, \tfalse\}$}}{
\For{\cycle $s = 1$ \KwTo $\ell\log({U}) + \log\log{n}$}{
\AtNode(\tcc*[f]{  two-choices [Round $1$]    }){$v$}{
\KwLet $u_1, u_2 \in N(v) $ uniformly at random\;
\eIf{\Color{$u_1$} $=$ \Color{$u_2$}}{
	\Color{$v$} $ \gets $ \Color{$u_1$}\;
	\Bit{$v$} $ \gets \ttrue$ \;
}{
	\Bit{$v$} $ \gets \tfalse$ \;
}
}
\For(\tcc*[f]{     bit-propagation     }){    round $t = 2$ \KwTo $2\log{|C|} + 2\log\log{n}$     }{
\AtNode(\tcc*[f]{       [Rounds $2$ to $2\log{|C|} + 2\log\log{n}$]     }){$v$}{
\KwLet $u \in N(v)$ uniformly at random\;
\If{\Bit{$u$}}{
	\Color{$v$} $ \gets $ \Color{$u$}\;
	\Bit{$v$} $ \gets \ttrue$ \;
}
}
}
}
}
\caption{Distributed Voting Protocol with One Bit of Memory. The variable $\ell$ is a large constant and $U$ is an upper bound on
$c_1/(c_1-c_2)$.
Since the
process runs in multiple \cycle{s} of length $\BigTheta{\log{k} + \log\log{n}}$ each, we
assume that every node has knowledge of $\ell\cdot U$, $n$ and $k$. }\label{alg:memory}
\end{algorithm}

If we assume that each node has knowledge of $\ifrac{n}{a}$, the run time can
be further reduced to $\BigO{\left(\log(c_1/\left(c_1-c_2\right))+\log\log{n}\right)\cdot\left(\log{(\ifrac{n}{a})} + \log\log{n}\right)}$, given
$\ifrac{n}{a}$ is smaller than $k^{\LittleO{1}}$.
We start our analysis with \autoref{lem:number-of-bits} where we derive a lower bound on the
number of bits set during the two-choices round. We will then use the results by Karp et al.\ \cite{KSSV00} to
argue that after the bit-propagation rounds the number of bits set is $n$ \whp, that is, 
the total number of bits set grows
until eventually every node has its bit set. 
Finally, we will prove in \autoref{lem:relative-colors} that the
relative number of bits set for \emph{large} colors remains close to the
initial (relative) value during the bit-propagation rounds. Together with the
growth of the total number of set bits, this leads to a growth of the imbalance
towards \A by at least a constant factor during each \cycle. 

We will use $x^{(i)}(t)$ to denote the random variable for the total number of nodes which have their
bit set in a round $t$ of phase $i$. When it is clear from the context, we simply use the notation $x(t)$.   Accordingly, $x^{(i)}(1)$ is the number of bits set
after the two-choices round of phase $i$. Additionally, we will use $x_{j}^{(i)}(t)$ to denote the
number of nodes of color \C{j} which have their bit set in a round $t$. Similarly as before, we simply write
$x_{j}(t)$ when the phase is clear from the context.  
Furthermore, when analyzing the growth in $x^{(i)}(t)$ and $x_{j}^{(i)}(t)$ with respect to $x^{(i)}(t-1)$ and $x_{j}^{(i)}(t-1)$, we will assume
that $x^{(i)}(t-1)$ and $x_{j}^{(i)}(t-1)$ are fixed.

\subsection{The Key Lemmas}
We start by showing that the initial number of
bits is well-concentrated around the expectation after the \TC sub-phase.
\begin{proposition} \label{lem:bits-concentrated}
For any color \C{j} with $c_j = \BigOmega{\sqrt{n\log{n}}}$ the number of nodes
of color \C{j} which have their bit set after the two-choices round is
concentrated around the expected value, that is,
\[x_j(1)=
\frac{c_j^2}{n}\left(1\pm\BigO{\ifrac{\sqrt{n\log{n}}}{  c_j  }}\right)
 \] \whp.
If $c_j = \BigO{\sqrt{n}\cdot\log^2{n}}$, then $x_j(1) = \BigO{\log^4{n}}$ \whp.
\end{proposition}

The following proposition bounds the growth of each opinion after one phase, that is after the \TC and \BP sub-phase.  

\begin{proposition} \label{lem:relative-colors}
Let $a'$ and $c_j'$ be the number of nodes of colors $\A = \C{1}$ and \C{j},
respectively, after the bit-propagation round. Let 
$T = 2(\log{k}+\log\log{n})$. Given $x(1)$ and assuming it is concentrated
around the expected value, we have after $T$ bit-propagation rounds \whp
\begin{align*}
a' &\geq \frac{a^2}{x(1)} \cdot \left(1 - \BigO{\frac{T\cdot\sqrt{n\log{n}}}{a}} \right) & \text{and}&&
c_j'  \leq \frac{c_j^2}{x(1)} \cdot \left(1 + \BigO{\frac{T\cdot\sqrt{n\log{n}}}{c_j}} \right) + k^3\cdot\log^7{n}\enspace .
\end{align*}
\end{proposition}

\subsection{Analysis}

In the following lemmas we analyze an arbitrary but fixed phase.

\begin{lemma} \label{lem:number-of-bits}
After the two-choices round, at least $\ifrac{n}{k}\cdot\left(1-\LittleO{1}\right)$ bits are set \whp.
\end{lemma}

\begin{proof}
The probability for one node to open connections to two nodes of the same color
is $p_\text{two-choices} = \sum_{\C{j}}\frac{c_j^2}{n^2}$. This probability is
minimized if all colors are of the same size $\ifrac{n}{k}$ and therefore
$p_{\text{min}} = \frac{1}{n^2}\cdot\sum_{\C{j}}\frac{n^2}{k^2} = \frac{1}{k}$.
Since all nodes open connections independently, the random variable for the
number of bits set after the two-choices round, $x(1)$, has a binomial
distribution with expected value at least $\Expected{x(1)} \geq \ifrac{n}{k}$.
Applying Chernoff bounds to $x(1)$ gives us
\begin{equation*}
\Probability{x(1) \leq \left(1-2\sqrt{\frac{k\log{n}}{n}}\right)\frac{n}{k}} \leq \Exp{-\frac{4kn\log{n}}{2kn}} = n^{-2} \enspace . \qedhere
\end{equation*}
\end{proof}

From the lemma above we obtain that we have at least
$x(1) = \ifrac{n}{k}\cdot\left(1-\LittleO{1}\right) = \BigOmega{\ifrac{n}{k}}$ bits
set after the first round \whp.

We are ready to prove \autoref{lem:bits-concentrated}, which states that the number of
bits is well-concentrated around the expectation for colors which are
\emph{large enough}.
\begin{proof}[Proof of \autoref{lem:bits-concentrated}]
Let \C{j} be an arbitrary but fixed color with $c_j>3\sqrt{n\log{n}}$. The
number of nodes of color \C{j} which have their bit set after the two-choices
round has a binomial distribution $x_j(1)\sim B(n,~\ifrac{c_j^2}{n^2})$ with
expected value $\Expected{x_j(1)}=\ifrac{c_j^2}{n} > 9\log{n}$. We apply
Chernoff bounds to $x_j(1)$ and obtain
\begin{equation*}
\Probability{|x_j(1) - \Expected{x_j(1)}| > 3\sqrt{\frac{\log{n}}{\Expected{x_j(1)}}} \cdot \Expected{x_j(1)}} \leq n^{-2} \enspace .
\end{equation*}
That is, we have
$|x_j(1)-\Expected{x_j(1)}|\leq3\sqrt{\log{n}\cdot\Expected{x_j(1)}}$ \whp.
Hence, \[x_j(1)=
\Expected{x_j(1)}\left(1\pm\BigO{\ifrac{\sqrt{\log{n}}}{\sqrt{\Expected{x_j(1)}}}}\right)= 
\frac{c_j^2}{n}\left(1\pm\BigO{\ifrac{\sqrt{n\log{n}}}{  c_j  }}\right)
 \] \whp.

The second statement can be shown in an analogous way.
\end{proof}

We now investigate the growth of $x(t)$ in the rounds following the \TC round.
\begin{lemma} \label{lem:KSSV00}
Assume that $x(1) \geq n/k\cdot(1-\LittleO{1})$.
After at most $T = 2(\log{k}+\log\log{n})$ bit propagation rounds, we
have $x(T) = n$, \whp.  Furthermore, \whp it holds that $1 \leq x(t+1)/x(t) \leq
2+\LittleO{1}$.
\end{lemma}
The proof follows from the results on rumor spreading in \cite{KSSV00}. While
in \cite{KSSV00} the authors analyze a combination of push and pull, an
elaborate analysis was conducted in \cite{Sch02} which considers the pull
operation separately. This latter analysis from \cite{Sch02} can be directly
applied to show our lemma.

We proceed by establishing bounds on the growth of the bits set of color \C{j}
after $t+1$ rounds, that is of  $x_j(t+1)$, for given $x_j(t)$.
\begin{lemma} \label{lem:colors-concentrated}
Let \C{j} be a color with at least $x_j(t) = \BigOmega{\log{n}}$ bits set in a
round $t$. Assume $x(t)$ and $x_j(t)$ are given and they are concentrated
around their mean.
Then we have \begin{equation*}
\Expected{x_j(t+1)| x(t) , x_j(t)} = x_j(t) + \frac{n-x(t)}{n}\cdot x_j(t) \enspace .
\end{equation*}
Furthermore, the number of nodes of color \C{j} which have their bit set in round $t+1$ is \whp concentrated
around the expected value such that
\begin{equation*}
x_j(t+1)=\Expected{x_j(t+1) | x_j(t) , x(t) }\cdot\left(1\pm\BigO{\frac{\sqrt{\log{n}}}{\sqrt{\Expected{x_j(t+1)| x(t) , x_j(t)}}}}\right) \enspace .
\end{equation*}
\end{lemma}

\begin{proof}
In the following, we will use $\operatorname{bit}_v(t)$ to denote the value of the bit of a node $v$ in a round $t$, where the value can be either \ttrue or \tfalse.
We consider the probability that $v$ has color \C{j} in round
$t+1$, given that $v$ has its bit set in round $t+1$. We have
\begin{align*}
\MoveEqLeft \Probability{v \in \C{j}(t+1) | \operatorname{bit}_v(t+1) = \ttrue , x_j(t) , x(t)} = {x_j(t)}/{x(t)} \enspace ,
\intertext{since}
\MoveEqLeft \Probability{v \in \C{j}(t+1) | \operatorname{bit}_v(t+1) = \ttrue,x_j(t), x(t)} \\
&= \frac{
\Probability{v \in \C{j}(t+1) \wedge \operatorname{bit}_v(t+1) = \ttrue| x_j(t), x(t)} }{\Probability{ \operatorname{bit}_v(t+1) = \ttrue | x_j(t), x(t) }} \\[0.5cm]
&= \frac{ \overbrace{\frac{x_j(t)}{n}\left(\frac{n-x(t)}{n}\right)}^{(i)} + \overbrace{\frac{x_j(t)}{n}}^{(ii)} }{ \underbrace{\frac{x(t)}{n}}_{(iii)} + \underbrace{{\left(1-\frac{x(t)}{n}\right)\frac{x(t)}{n} }}_{(iv)} } 
= \frac{ x_j(t) }{x(t)} \cdot \frac{1-\frac{x(t)}{n} + 1}{1 + 1 - \frac{x(t)}{n}} 
\enspace .
\end{align*}
In above equation, the probability for a node to have color \C{j}
and the bit set in round $t+1$ is computed as follows.
\begin{enumerate}[label={(\roman*)}]
\item is the probability that a node has color \C{j} and the bit set at time $t$ and selects a node without a bit set
\item is the probability that a node chooses another node which has color \C{j} and the bit set
\item is the probability for choosing a node with a set bit
\item is the probability for choosing a node without the bit set which selects another node with the bit set
\end{enumerate}

Consequently, the number of nodes which have color \C{j} in the next round has
expected value
$\mu = \Expected{x_j(t+1)|x(t+1) , x_j(t) , x(t)} = {x_j(t)\cdot x(t+1)}/{x(t)}$. 
We apply Chernoff bounds to $x_j(t+1)$ and obtain
\begin{equation*}
\Probability{\left.|x_j(t+1) - \mu| > 3\sqrt{\frac{\log{n}}{\mu}} \cdot \mu \right| x_j(t) , x(t) , x(t+1)} \leq n^{-2} \enspace .
\end{equation*}
Assuming $x(t)$ fulfills \autoref{lem:number-of-bits}, we have \cite{KSSV00} \begin{equation*}
x(t+1) = \Expected{x(t+1)|x(t)}\cdot\left(1\pm\BigO{{\sqrt{k\log{n}}}/{\sqrt{n}}}\right) \enspace , \end{equation*} and therefore we obtain the lemma.
\end{proof}

We are ready to prove \autoref{lem:relative-colors}.

\begin{proof}[Proof of \autoref{lem:relative-colors}]
Let $a_i = x_1(i)$ be a sequence of random variables for the number of nodes of
color \A which have their bit set in round $i$.
In the following proof, whenever we condition on $a_j$ or $x(j)$ for any $j$, we assume that they
are concentrated around their mean according to \autoref{lem:KSSV00}, \autoref{lem:bits-concentrated}, and \autoref{lem:colors-concentrated}. 

According to \autoref{lem:colors-concentrated} we know that
\begin{equation*}
\Expected{a_{i+1}|a_{i} , x(i+1) , x(i)} = \frac{x(i+1)}{x(i)} \cdot a_{i} 
\enspace .
\end{equation*}
Note that $\Expected{a_{i+1}|a_i} \geq a_i$.
Therefore we have
\begin{equation*}
\Probability{\left. a_{i+1} < \frac{x(i+1)}{x(i)} \cdot a_{i} \cdot \left( 1 - \frac{3\sqrt{\log{n}}}{\sqrt{a_{i}}} \right) \right| a_{i} , x(i-1) , x(i)} \leq n^{-2} \enspace .
\end{equation*}

The total number of bits set in the round $i+1$, given the total number of bits
in round $i$, is independent of the color distribution among these nodes in
round $i$, that is, for any $\beta \leq \gamma$ it holds for any $\alpha$ that
\begin{equation*}
\Probability{x(i+1) = \alpha | x_j(i) = \beta , x(i) = \gamma } = \Probability{x(i+1) = \alpha | x(i)= \gamma } \enspace .
\end{equation*}
We therefore have for any $\tau > i$
\begin{equation*}
\Probability{\left. a_{i+1} < \frac{x(i+1)}{x(i)} \cdot a_{i} \cdot \left( 1 - \frac{3\sqrt{\log{n}}}{\sqrt{a_{i}}} \right) \right| a_{i} , x(1) , \dots , x(\tau)} \leq n^{-2} \enspace .
\end{equation*}

\noindent The equation above means that the distribution of the colors among the nodes with the bit set
at time $i+1$, given $x(1) \dots x(i+1)$, is independent of the number of nodes with the bit set
at times $i+2, \dots, \tau$.

Recall that, given $a_1$, $a_{i} = \BigOmega{a_1}$ \whp and therefore
we have for given $a_1$, $a_{i}$, $x(i-1)$, $x(i)$, and a constant $\zeta$ \whp
\begin{equation} \label{eq:ind-1}
a_{i+1} \geq \frac{x(i+1)}{x(i)} \cdot a_{i} \cdot \left( 1 - \zeta\cdot\frac{\sqrt{\log{n}}}{\sqrt{a_1}} \right) \enspace .
\end{equation}
Define $T = \BigO{\log{(\ifrac{n}{a})} + \log\log{n}}$ such that $x(T) = n$ \whp
according to \cite{KSSV00}. We now show by induction that, given
$a_1$, $x(1),\dots,x(T)$, and a constant $\zeta$,
\begin{equation} \label{eq:zeta}
a_T \geq \frac{x(T)}{x(1)} \cdot a_1 \cdot \left(1 - \zeta\cdot\frac{\sqrt{\log{n}}}{\sqrt{a_1}} \right)^T
\end{equation}
\whp. The base case for round $t=1$ obviously holds. For the step from $t$ to
$t+1$ we use \eqref{eq:ind-1} as follows.
\begin{align*}
a_{t+1} &\stackrel{\eqref{eq:ind-1}}{\geq} \frac{x(t+1)}{x(t)} \cdot a_t \cdot \left(1 - \zeta \cdot 
\frac{\sqrt{\log{n}}}{\sqrt{a_1}}\right) \\ &
\stackrel{\text{IH}}{\geq} \frac{x(t+1)}{x(t)} \cdot \frac{x(t)}{x(1)} \cdot a_1 \cdot \left(1 - \zeta\cdot\frac{\sqrt{\log{n}}}{\sqrt{a_1}} \right)^t \cdot \left(1 - \zeta\cdot\frac{\sqrt{\log{n}}}{\sqrt{a_1}} \right) \\
&\geq \frac{x(t+1)}{x(1)} \cdot a_1 \cdot \left(1 - \zeta\cdot\frac{\sqrt{\log{n}}}{\sqrt{a_1}} \right)^{t+1}
\end{align*}
This concludes the induction. We apply the Bernoulli inequality to \eqref{eq:zeta} and obtain
\begin{equation} \label{eq:zeta2-xxx}
a_T \geq \frac{x(T)}{x(1)} \cdot a_1 \cdot \left(1 - \zeta\cdot\frac{T\cdot\sqrt{\log{n}}}{\sqrt{a_1}} \right) \enspace .
\end{equation}
We use the result from \autoref{lem:bits-concentrated} for $a_1$ in
\eqref{eq:zeta2-xxx} and obtain
\begin{align*}
a' &\geq \frac{n}{x(1)} \cdot a_1 \cdot \left(1 - \zeta\cdot\frac{T\cdot\sqrt{\log{n}}}{\sqrt{a_1}} \right) \\
&\geq \frac{n}{x(1)} \cdot \frac{a^2}{n}
	\cdot \underbrace{\left(1 - \zeta\cdot\frac{T\cdot\sqrt{\log{n}}}{\sqrt{a_1}} \right)}_{\text{\textsc{(i)}}}
	\cdot \underbrace{\left(1 - \frac{3\sqrt{\log{n}}\cdot\sqrt{n}}{a}\right)}_{\text{\textsc{(ii)}}} \enspace ,
\intertext{where the second expression in parentheses, \textsc{(ii)}, is
asymptotically dominated by the first one, \textsc{(i)}. Therefore, there is a $\zeta'$ such that}
a' &\geq \frac{a^2}{x(1)} \cdot \left(1 - \zeta'\cdot\frac{T\cdot\sqrt{\log{n}}}{\sqrt{a_1}} \right) \enspace . \numberthis \label{eq:bound-a}
\end{align*}

A similar upper bound can be computed for any \emph{large} color.
Let \C{j} be an arbitrary but fixed color and assume that $c_j \geq \sqrt{n}\cdot\log^2{n}$.
We have, by \autoref{lem:bits-concentrated}, that \whp 
\[x_j(1) = \frac{c_j^2}{n}  \left(1 \pm \zeta\cdot\frac{\sqrt{n\log{n}}}{c_j} \right) \enspace . \] 
By using analogous arguments as for color \A we obtain \whp
\begin{equation*}
 c_j' \leq \frac{c_j^2}{x(1)} \cdot \left(1 + \zeta'\cdot\frac{T\cdot\sqrt{\log{n}}}{\sqrt{x_j(1)}} \right) \enspace .  
\end{equation*}
If otherwise $c_j < \sqrt{n}\cdot\log^2{n}$, we have by \autoref{lem:bits-concentrated} that \whp
\[x_j(1) = \BigO{\log^4 n} \enspace .  \] 
We have \whp
\[x_j(t+1) \leq (1+\LittleO{1})\frac{x(t+1)}{x(t)} x_j(t) + \log^2 n \enspace .\]
By \autoref{lem:KSSV00} we have that ${x(t+1)}/{x(t)} \leq 2(1+\LittleO{1}) $ \whp.
Thus, since there are $T = 2(\log{k}+\log\log{n})$ many rounds, taking union bound, 
we have  \whp
\[
c_j' =\BigO{k^{2\log(1+\LittleO{1})}\cdot\log^{4+2\log(1+\LittleO{1})}{n}} \enspace .
\]
 Thus for any $\C{j}$ we have \whp
 \begin{equation}\label{eq:bound-b}
 c_j' \leq \frac{c_j^2}{x(1)} \cdot \left(1 + \zeta'\cdot\frac{T\cdot\sqrt{\log{n}}}{\sqrt{x_j(1)}} \right) +\BigO{k^{3}\cdot\log^{7}{n}}  \enspace .  
\end{equation}
 
Taking all contributions into consideration, we
observe that there always exists a constant $\zeta'$ such that \eqref{eq:bound-a} and
\eqref{eq:bound-b} are satisfied.
\end{proof}

We are now ready to put all pieces together and prove our main theorem, \autoref{thm:memory},
which is restated as follows.

\begin{restate}[thm:memory]
\theoremmemory
\end{restate}

\begin{proof}
Assume $x(1)$ is given and concentrated around its expected value.
Recall that in the statement of \autoref{thm:memory} we assume $a-b \geq z\cdot\sqrt{n \log^3{n}}$. 

For the following calculations, we assume that $b \geq \sqrt{n}\log^2 n$. 
Let $\delta > 0$ be a constant. We distinguish the following two cases.
\paragraph{Case 1: $a < \left(1+\delta\right) b$.}
Let
$T = 2(\log{k}+\log\log{n})$.
From the bounds on $a'$ and $b'$ from \autoref{lem:relative-colors} we obtain the following inequality, which holds \whp.
\begin{align*}
a' - b' &\geq \frac{a^2-b^2}{x(1)}- \frac{\zeta \cdot T\cdot\sqrt{\log{n}}}{x(1)}\cdot \left(\frac{a^2}{\sqrt{a_1}} + \frac{b^2}{\sqrt{b_1}}\right)\\
&\geq \frac{a^2-b^2}{x(1)} - \frac{2 \cdot \zeta \cdot T\cdot\sqrt{\log{n}}}{x(1)}\cdot \frac{a^2}{\sqrt{a_1}} \\
&\geq \frac{a-b}{x(1)} \cdot \left( \left(a+b\right) - \frac{2 \cdot \zeta \cdot T\cdot\sqrt{\log{n}}}{a-b}\cdot \frac{a^2}{\sqrt{a_1}}\right) \\
\intertext{(using $a_1 = \ifrac{a^2}{n}\cdot\left(1\pm \LittleO{1}\right)$ \whp according to \autoref{lem:bits-concentrated})}
&\geq \frac{a-b}{x(1)} \cdot\left(\left(a+b\right) - \frac{2\cdot\zeta \cdot T\cdot\sqrt{\log{n}}\cdot a^2\cdot\sqrt{n}}{\left(a-b\right)\cdot a\cdot\left(1-\LittleO{1}\right)} \right)
\intertext{(using $a-b \geq z\cdot\sqrt{n\log^3{n}}$)}
&\geq \frac{a-b}{x(1)} \cdot\left(\left(a+b\right) - \frac{2\cdot\zeta \cdot a}{z\cdot \left(1-\LittleO{1}\right)} \right)
\intertext{Now if $z$ is large enough, we obtain for a small positive constant $\varepsilon = \varepsilon(z)$ that}
a'-b' & \geq \left(a-b\right)\cdot\left(\frac{a\cdot \left(1 - \varepsilon\right) + b}{x(1)}\right) \enspace . \numberthis \label{eq:memory-result}
\end{align*}

We combine  the bound on $b'$ of \autoref{lem:relative-colors} with \eqref{eq:memory-result} and 
obtain \whp
\begin{align*}
\frac{a'-b'}{b'} & \geq (a-b)\cdot\left(\frac{a\cdot\left(1 - \varepsilon\right) + b}{x(1)}\right) \cdot \frac{x(1)}{b^2\cdot\left(1 + \LittleO{1}\right)}\\
&= \frac{a-b}{b}\cdot \left( \frac{a\cdot\left(1-\varepsilon\right)+b}{b\cdot\left(1 + \LittleO{1}\right)}\right)\\
&\geq \frac{a-b}{b}\cdot \left( \frac{b\cdot\left(2-\varepsilon\right)}{b\cdot\left(1 + \LittleO{1}\right)}\right)\\
&= \frac{a-b}{b}\cdot \left( 1+\varepsilon'\right) 
\intertext{where $\varepsilon' > 0$ is a positive constant. Let $a^{(i)}$ and $b^{(i)}$ denote the number of nodes of color \A and \B, respectively, after $i$ \cycle{s}. After $i = \log_{1+\varepsilon'}\left(a/(c_1-c_2)\right)$ \cycle{s} we have \whp}
\frac{a^{(i)}-b^{(i)}}{b^{(i)}} &\geq \frac{a-b}{b}\cdot \left( 1+\varepsilon'\right)^{\log_{1+\varepsilon'}\frac{a}{a-b}}\\
&= \frac{a-b}{b}\cdot\frac{a}{a-b}
\geq 1 \enspace .
\end{align*}
We therefore get after $i$ \cycle{s} that $a^{(i)}-b^{(i)} \geq b^{(i)}$ and thus ${a^{(i)}}/{b^{(i)}} \geq 2$.
\paragraph{Case 2: $a \geq \left(1+\delta\right) b$.}
We consider 
the ratio between $a'$ and $b'$ and show  a quadratic growth w.r.t. $a^2/b^2$. 
We apply \autoref{lem:bits-concentrated} and \autoref{lem:relative-colors} to derive

\begin{equation*}
\frac{a'}{b'} \geq \frac{\frac{a^2}{x(1)} \cdot \left(1 - \zeta\cdot\frac{\log^{\frac{3}{2}}{n}}{\sqrt{a_1}} \right) }{\frac{b^2}{x(1)} \cdot \left(1 + \zeta\cdot\frac{\log^{\frac{3}{2}}{n}}{\sqrt{b_1}} \right)}
= \frac{a^2}{b^2}\cdot \frac{1-\LittleO{1}}{1+\LittleO{1}}
\geq \frac{a^2}{b^2} \cdot \left(1 - \LittleO{1}\right),
\end{equation*}
where $\zeta$ is a suitable constant.

\paragraph{Putting everything together.}
Note that if $a < \left(1+\delta\right)b$ then after $i =
\log_{1+\epsilon'}(a/(a-b))$ \cycle{s} we have ${a^{(i)}}/{b^{(i)}} \geq 2$.
From here on the second case applies as long as  $b \geq \sqrt{n }\log^2 n$.
Observe that after $\BigO{\log\log{n}}$ \cycle{s}, every color except for \A
drops below $\sqrt{n}\cdot\log^2 n$.

Let \C{j} be an arbitrary but fixed color of size $c_j$. If $c_j$ is smaller
than $\sqrt{n}\cdot\log^2{n}$, we have by \autoref{lem:relative-colors} at the
end of the \BP sub-phase at most $\BigO{k^{3}\cdot\log^7{n}}$ nodes of color
\C{j}, \whp. Since we have $k \leq n^\epsilon$, in the next two-choices phase
this color will disappear with probability $1-1/n^{\BigOmega{1}}$. If \C{j}
does not disappear, the same argument applies, since $c_j\leq
\sqrt{n}\cdot\log^2 n$. Therefore, after a constant number of phases $\C{j}$
disappears \whp.

Thus, once \A is the only color having a support of more than $\sqrt{n}\cdot\log^2 n$,
all other colors will vanish \whp after additional $\BigO{T}$ rounds and \A
will be the only remaining color. This concludes the proof.
\end{proof}

\paragraph{Room for Improvement.}

The bound on the  plurality consensus time can be further improved to  $\BigO{\left(\log(c_1/\left(c_1-c_2\right))+\log\log{n}\right)\cdot\log{k}}$, which is of interest for cases where $k=\LittleO{\log n}$.
This can be achieved by having shorter \BP sub-phases in which not all nodes but a large fraction of nodes set their bit.  

\section{Analysis of the Asynchronous Algorithm}\label{sect:asynchronous-analysis}

\newcommand\var[1]{\ensuremath{\text{\normalfont\texttt{#1}}}\xspace}

\newcommand{\tp}[2]{\ensuremath{\tau_{\text{\textsc{\MakeLowercase{#1}}}#2}}\xspace}

\def\tSG#1{\tp{SG}{#1}}
\def\tBP#1{\tp{BP}{#1}}
\def\tTC{\tp{TC}{}}
\def\tMETA#1{\tp{End}{#1}}

\def\tFIN{\tp{T}{}}
\def\tJUMP{\tp{jump}{}}
\def\tSET{\tp{}{set}}
\def\tM#1{\tp{M}{#1}}
\def\bulk{\ensuremath{\mathcal{S}}\xspace}

We now introduce our asynchronous protocol to solve plurality consensus. In the
sequential asynchronous model we assume that a sequence of discrete time steps
is given, where at each time step one node is chosen uniformly at random to
perform its tick. Recall that the key to the speed of the synchronous algorithm
(\onebit) is the combination of the two-choice process with an information
dissemination process. However, this interweaving of these processes requires
that the nodes execute the sub-phases simultaneously. While this is trivially
the case in the synchronous setting, it is extremely unlikely in the
asynchronous setting, since the numbers of ticks of different nodes may differ
by up to $\BigO{\log n}$. Therefore, any attempt to reach full synchronization
is futile if one aims for a run time of $\BigO{\log n}$.

To overcome this restriction, we adopt the following weaker notion of
synchronicity. At any time we only require a $(1-\LittleO{1})$ fraction of the
nodes to be \emph{almost synchronous}. This relaxes full synchronicity in three
ways: First, nodes are only almost synchronous, meaning that for any two nodes
their working times may differ by up to $\Delta = \BigTheta{\log n/\log \log
n}$. Secondly, we allow $\LittleO{n}$ nodes to be poorly synchronized. Finally,
we require this to hold only \whp.

The above notion does not require the nodes to synchronize actively per se,
since their number of ticks is to some extent concentrated even without active
synchronization. However, it turns out that without synchronizing perpetually,
the number of poorly synchronized nodes in each phase will become larger than
the initial bias towards the plurality opinion $c_1-c_2$ and could therefore
influence the consensus significantly. We thus actively synchronize nodes at
the end of each phase to decrease the fraction of poorly synchronized nodes
such that their number is in $\LittleO{c_1 - c_2}$, resulting in a negligible
influence of those nodes.

Once several technical challenges are resolved, the resulting weak
synchronicity allows us to reuse the high-level structure of the synchronous
algorithm (\onebit). As in the synchronous case, the asynchronous protocol
consists of one \TC sub-phase and one \BP sub-phase, the latter of which
propagates the choices of the \TC phase to all nodes in the network. In
addition to these sub-phases we have a third sub-phase in which we synchronize
nodes.

After executing the first two sub-phases, the relative difference between $\A$
and any opinion $\C{j}\neq A$ increases quadratically and thus we only require
$\BigO{\log \log n}$ such phases. Each of the sub-phases has a length of
$\BigO{\log n / \log \log n}$, amounting to a total run-time of $\BigO{\log
n}$. While superficially the asynchronous version looks very similar to the
synchronous protocol (\onebit), the analysis differs greatly from the
synchronous case, in both approach and technical execution.

\subsection{The Asynchronous Protocol}

\begin{float}[t!]
\begin{algorithm}[H]
\def\toe{\ensuremath{\ifrac{T}{10}}}
\SetAlgoVlined
\SetKw{KwLet}{let}
\SetKwProg{Algorithm}{Algorithm}{ (Part 1)}{end}
\SetKwFunction{Asynchronous}{asynchronous}
\SetKwFunction{xSamples}{samples}
\SetKwFunction{xMedian}{median}
\SetKwFunction{xColor}{color}
\SetKwFunction{xTemp}{intermediate}
\SetKwFunction{xBit}{bit}
\SetKwFunction{xWorkingTime}{workingtime}
\SetKwFunction{xRealTime}{realtime}
\def\Samples#1{\xSamples{\ensuremath{#1}}}
\def\Median#1{\xMedian{\ensuremath{#1}}}
\def\Color#1{\xColor{\ensuremath{#1}}}
\def\Temp#1{\xTemp{\ensuremath{#1}}}
\def\Bit#1{\xBit{\ensuremath{#1}}}
\def\WorkingTime#1{\xWorkingTime{\ensuremath{#1}}}
\def\RealTime#1{\xRealTime{\ensuremath{#1}}}
\SetKw{KwLet}{let}
\SetKw{KwAnd}{and}
\SetKw{KwNext}{continue with}
\Algorithm{\Asynchronous{node $v$}}{
\KwLet $T = \kappa \cdot\log{n}/\log\log{n}$\;
\KwLet $t = \WorkingTime{v} \mod T$\;
\begin{minipage}[t]{0.4\textwidth}
\uIf{$t = \toe$}{
	\KwLet $u_1, u_2 \in N(v)$ uniformly at random\;
	\eIf{$\Color{u_1} = \Color{u_2}$}{
		$\Temp{v} \gets \Color{u_1}$\;
	}{
		$\Temp{v} \gets \tnull$\;
	}
}
\uElseIf{$t = 2\cdot\toe$}{
	\eIf{$\Temp{v} \neq \tnull$}{
		$\Color{v} \gets \Temp{v}$\;
		$\Bit{v} \gets \ttrue$\;
	}{
		$\Bit{v} \gets \tfalse$\;
	}
}
\uElseIf{$t \in \left[ 3\cdot\toe,~7\cdot\toe \right]$}{
	\If{$\Bit{v} = \tfalse$}{
		\KwLet $ u \in N(v)$ uniformly at random\;
		\If{$\Bit{u} = \ttrue$}{
			$\Bit{v} \gets \ttrue$ \smash{asdf}\;
			\parbox[t][0pt][t]{0pt}{ \vspace{-5pt} \hspace{-49pt}\colorbox{white}{\vdots}}
			$\Color{v} \gets \Color{u}$\;
		}
	}
}
\end{minipage}\hfill
\begin{minipage}[t]{0.5\textwidth}
\vdots
\uElseIf{$t \in \left[8\cdot\toe,~9.5 \cdot\toe \right]$}{
	increase all values in $\Samples{v}$ by $1$\;
	\If{$t \in \left[8\cdot\toe,~8\cdot\toe + \log^3\log{n}\right]$}{
		\KwLet $ u \in N(v)$ uniformly at random\;
		$ \Samples{v} \gets \Samples{v} \cup  \left\{\RealTime{u}\right\}$\;
	}
	\If{$t = 9.5 \cdot\toe$ \KwAnd $\Samples{v} \neq \emptyset$}{
		$\WorkingTime{v} \gets \Median{ \Samples{v} }$\;
		$\Samples{v} \gets \emptyset$\;
	}
}
\Else{
	do nothing\;
}

\bigskip

$\RealTime{v} \gets \RealTime{v} + 1$\;
$\WorkingTime{v} \gets \WorkingTime{v} + 1$\;
\If{$\WorkingTime{v} \geq \kappa \cdot \ell \cdot \log{n}$}{
	 \KwNext \autoref{alg:endgame}\;
}
\end{minipage}
}
\caption{Part 1 of the asynchronous protocol to solve plurality consensus. Both
variables \var{realtime} and \var{workingtime} are initialized to $0$, and
\var{samples} is initially the empty set. The variables $\kappa$ and $\ell$
denote large constants. The goal of the algorithm is to increase the plurality
opinion \A such that $a \geq (1-\teps{Part1})\cdot n$ for a small constant
\teps{Part1}.}
\label{alg:asynchronous}
\end{algorithm}

\begin{figure}[H] \centering
\includegraphics{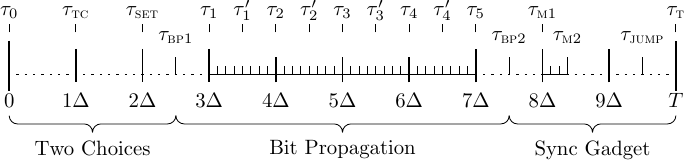} \vspace{-\baselineskip}
\caption{graphical representation of one phase of \autoref{alg:asynchronous}.
Each phase consists of $T = 10 \cdot \Delta$ ticks.}
\label{fig:asynchronous}
\end{figure}
\end{float}

Our asynchronous protocol consists of two parts, Part~1 defined in
\autoref{alg:asynchronous} later in this section and Part~2 defined in
\autoref{alg:endgame} in \autoref{sec:endgame}. In these formal definitions, we
specify the operations that each node performs when selected to tick. The goal
of the first part is to increase the number of nodes of color \A to at least $a
\geq \left(1 - \teps{Part1}\right)\cdot n$ for some small constant
\teps{Part1}. Once the execution of the first part has finished, the nodes
execute a simple two-choices algorithm in an asynchronous manner. We will show
that after the second part, \A wins \whp. Our main contribution is the analysis
of the first part. For the sake of completeness, we formally analyze the second
part in \autoref{sec:endgame}.

In contrast to the formal definitions, it is more convenient and instructive to
represent the algorithm executed by each node in a graphical way. This
graphical representation for a single phase of the first part is shown in
\autoref{fig:asynchronous}. In this graphical representation, the instructions
are drawn on a line from left to right, starting with the first instruction at
the left endpoint.

As in the synchronous case, the asynchronous algorithm operates in multiple
phases. Each of these phases is split into three sub-phases. Each sub-phase consists of multiple blocks of length $\Delta$ each.
During these sub-phases, according to \autoref{alg:asynchronous}, there are
multiple blocks of instructions where nodes for a long time literally \emph{do nothing}. These
do-nothing-blocks are used, in combination with the following result on synchronicity, to
ensure that a large fraction of nodes executes critical instructions at almost
the same time. That is, for a large fraction of nodes we will show that these
nodes execute instructions as if they were bulk synchronized, which they
clearly are not.

The first phase is the \TC sub-phase, which consists of two instructions, the
\TC step and the commit step. In the \TC step, every node samples two neighbors
uniformly at random. If and only if these neighbors' colors coincide, the node
sets an \var{intermediate} color to the neighbors' colors. In the commit step,
nodes change their color if they have their intermediate color set and then set
their bit accordingly. The second phase is the \BP sub-phase, which closely
resembles the synchronous counter part. Finally, in the third phase, all nodes
execute the so-called \emph{\SG}. In this gadget, nodes adjust their
\emph{working time} in order to synchronize. Our perpetual synchronization
mechanism is described after the following definitions.

For the analysis of the asynchronous algorithm we will use the following
notation and definitions.

\paragraph{Definitions.}
Let $\kappa$ and $\ell$ denote sufficiently large positive constants. We refer
to a series of $n$ consecutive time steps as a \emph{period}, and we combine
$T=\kappa\cdot\log{n}/\log\log{n}$ periods to a \emph{phase}. The first part of
the asynchronous protocol consists of $\ell\cdot\log\log{n}$ phases.
Intuitively, a period is the number of time steps during which each node ticks
in expectation once. We define a \emph{reference point} $\tau$ to be a time
step which marks the end of a period $\tau$. In particular, at reference point
$\tau$ there have been $\tau \cdot n$ time steps, and each node has ticked in
expectation $\tau$ times.

{ \em \begin{itemize}
\item
Let $T_v(t)$ denote the random variable for the \emph{real time}, the number
of ticks of node $v$ after the first $t \cdot n$ time steps. That is, $T_v(t)$
denotes the number of times $v$ was scheduled during the first $t\cdot n$
ticks.

\item
Let $T'_v(t)$ denote the random variable for the \emph{working time}, the
current instruction counter of node $v$ after the first $t \cdot n$ time
steps. Note that $T'_v(t)$ can differ from $T_v(t)$ since the working time
is adjusted with the goal of synchronization in \autoref{alg:asynchronous}.
\end{itemize}
}
At the beginning of the algorithm, both, the real time and the working time are
initialized to~$0$. Since at each time step one node is chosen to tick
independently and uniformly at random, $T_v(\tau)$ has a binomial distribution
$T_v(\tau) \sim B(\tau\cdot n,~\ifrac{1}{n})$ with expected value
$\Expected{T_v(\tau)} = \tau$.
It will prove convenient to regard a reference point as
the one instruction in the algorithm which would be executed in the
corresponding period if every node ticked exactly once in every period.

\paragraph{Weak Perpetual Synchronization.}

In the asynchronous algorithm, when a node is selected to tick, all operations
are performed based on the node's current working time. In contrast, the real
time of a node is used to always the total number of ticks performed so far by
this node. In the \SG, the working time $T'_v$ of a node $v$, denoted as
\var{workingtime} in \autoref{alg:asynchronous}, is adjusted as follows.

The \SG consists of a sampling sub-phase $[\tM1, \tM2]$ and a jump step \tJUMP.
The sampling sub-phase of the \SG consists of $\log^3 \log n$ ticks. During
these ticks, every node samples a neighbor uniformly at random and collects the
real time $T_u$ of the sampled neighbor $u$. Additionally, the node increments
all real times sampled so far by $1$ until the jump step is executed. At the
jump step, the node sets its working time to the median of the samples.

During the entire phase, according to \autoref{alg:asynchronous}, there are
multiple blocks of instructions where nodes literally \emph{do nothing}. These
blocks are used, in combination with the following result on synchronicity, to
ensure that a large fraction of nodes executes critical instructions at almost
the same time. That is, for a large fraction of nodes we will show that these
nodes execute instructions as if they were bulk synchronized, which they
clearly are not.

\subsection{The Key Lemmas}

The use of the \SG and the following definition of $\Delta$-closeness allow us
to show \autoref{prop:bulk} which forms the basis for our adaption of the
synchronous protocol to the asynchronous setting.

\begin{definition} \label{def:delta-close}
We say a node is $\Delta$-close to a reference point $\tau$ w.r.t.\ the real
time $T_v$ or the working time $T'_v$, if $\left|T_v(\tau) - \tau\right| \leq
\Delta$ or $\left|T'_v(\tau) - \tau\right| \leq \Delta$, respectively. If we
say a node is $\Delta$-close without specifying a reference point, we mean that
it is $\Delta$-close to the expected number of ticks.
\end{definition}

\begin{proposition} \label{prop:bulk}
Let \bulk be set of \emph{synchronized nodes} that are $(\Delta/2)$-close
w.r.t.\ the working time throughout the entire process. \Whp,
$|\bulk| \geq n \cdot \left( 1 - \Exp{-8\log{n}/\log\log{n}}\right)$.
\end{proposition}
The proof idea is as follows.
We first observe that roughly $ n \cdot \left( 1 - \Exp{-\log{n}/\log^2\log{n}}\right)$
nodes are $(\Delta/16)$-close throughout the execution of the algorithm. 
As argued before, the resulting  number of poorly synchronized nodes is too large and could tip the balance.
Furthermore, we show, by careful induction, that thanks to the
perpetual synchronization in each phase, a large fraction $f=\left( 1 - \Exp{-9\log{n}/\log\log{n}}\right)$ of the nodes which
were  $(\Delta/2)$-close throughout the first $i$ phases, will 
remain $(\Delta/2)$-close in phase $i+1$:
(i) a fraction  $f$ of these nodes will tick equally often in each interval in this phase, up to an error of $\Delta/16$, 
   and (ii) among these nodes again a fraction $f$ will adapt their working time by selecting the median of a sample of nodes. That median will be $(\Delta/16)$-close.
   Accounting for numerous other sources of error we obtain overall $(\Delta/2)$-closeness for a large fraction of nodes.

Equipped with \autoref{prop:bulk} we analyze the \TC and \BP sub-phases.
\autoref{prop:two-choices} and \autoref{prop:bit-propagation} form the
asynchronous counter parts of \autoref{lem:bits-concentrated} and
\autoref{lem:relative-colors}, subject to a subtle difference: Instead of
describing the distribution of colors after every \TC and \BP sub-phase, we
restrict ourselves to the distribution of colors among the well-synchronized
nodes in \bulk. In fact, throughout the analysis, we assume for all other nodes
in $(V\setminus \mathcal{S})$ the worst-case. However, based on the \SG and
\autoref{prop:bulk}, their number is small enough such to prevent them from
tipping the balance.

Our next key-lemma is \autoref{prop:two-choices} which establishes that the number of nodes which pick up a bit for
color \C{j} is \whp concentrated around the expectation. 

Analogously to the synchronous case, we consider in the following definitions
and propositions an arbitrary but fixed phase of \autoref{alg:asynchronous}.
Let $\hat c_{j}(\tau)$ denote the number of nodes belonging to \bulk having color
\C{j} at reference point $\tau$, that is, at time step $\tau \cdot n$. Let
furthermore $x_j(\tau)$ denote the set of nodes belonging to \bulk having color
\C{j} and having their bit set at reference point $\tau$ and let finally
$x(\tau)=\sum_j x_j(\tau)$.

\begin{proposition} \label{prop:two-choices}
Assume $|\bulk| \geq n \cdot \left( 1 - \Exp{-8\log{n}/\log\log{n}}\right)$.
Let \C{j} be an arbitrary but fixed color. \Whp, the number of nodes in \bulk
having a bit set for color $\C{j}$ after the \TC sub-phase at reference point
\tBP{1} is bounded as follows.
\begin{align*}
x_1(\tBP{1}) &\geq\frac{ \hat c_j(\tau_0)^2}{n}\left(1 - \LittleO{1}\right) 
& \text{and}&&
x_i(\tBP{1}) &\leq\frac{\hat c_j(\tau_0)^2}{n}\left(1 + \LittleO{1}\right) + \BigO{n^{1-{14}/{\log\log n}}} \enspace .
\end{align*}
\end{proposition}

Building on the concentration of bits given by \autoref{prop:two-choices} at
\tBP{1}, the following proposition  bounds the number of nodes of each color
after the \BP sub-phase at \tBP{2}. As before, we only characterize those nodes
which are part of \bulk.
\begin{proposition} \label{prop:bit-propagation}
Assume $|\bulk| \geq n \cdot \left( 1 - \Exp{-8\log{n}/\log\log{n}}\right)$.
Let \C{j} be an arbitrary but fixed color. \Whp, the number of nodes in \bulk
of color $\C{j}$ after the \BP sub-phase is bounded as follows.
\begin{align*}
\hat c_1(\tp{BP}{2}) &\geq \frac{\hat c_1(\tau_0)^2}{ x(\tp{BP}{1}) } \cdot \left(1 - \LittleO{1} \right) & \text{and}&&
\hat c_j(\tp{BP}{2}) &\leq \frac{ \hat c_j(\tau_0)^2 }{ x(\tp{BP}{1}) } \cdot \left(1 + \LittleO{1} \right) + \BigO{n^{1-{4}/{\log\log n}}} \enspace .
\end{align*}
\end{proposition}
In the proof we analyze the \BP by the means of the Pólya urn p rocess. In
particular, we show that the fraction of nodes supporting each color \C{j}
remains concentrated throughout the \BP sub-phase. The proofs can be found in
\autoref{sec:proofas3}, \autoref{sec:proofas1}, and \autoref{sec:proofas2},
respectively.

\subsection{Concentration of the Clocks: Proof of \autoref{prop:bulk}}\label{sec:proofas3}

In the following we show that throughout the entire process there do not exist
nodes which perform more than $\BigO{\log{n}}$ ticks, \whp. 

\begin{observation} \label{obs:max-number-of-ticks}
For any reference point $\tau$ we have that the working time
of any node is bounded by the minimum and maximum real times, that is,
for all $u\in V$ and $\tau \in \mathbb{N}$ we have 
\begin{equation} \label{eq:bound-on-working-time}
T'_v(\tau)\in \left[\min_{u\in V} T_u(\tau), \max_{u\in V} T_u(\tau)\right] \enspace .
\end{equation}
Let $\mathfrak{T}$ denote the total number of time steps until all nodes have
completed the execution of Part~1 of the asynchronous protocol defined in
\autoref{alg:asynchronous} w.r.t.\ their working time. \Whp, we have
\begin{equation} \label{eq:time-until-completion}
\mathfrak{T} \leq \ifrac{3}{2}\cdot \kappa \cdot \ell \cdot n\log{n} \enspace .
\end{equation}
Furthermore, we have \whp that
\begin{align} \label{eq:max-number-of-ticks}
\max_{v\in V}\left\{T_v( \mathfrak{T}  )\right\} < 2\cdot\kappa\cdot\ell\cdot \log{n} && \text{ and } && \max_{v\in V}\left\{T'_v( \mathfrak{T}  )\right\} < 2\cdot\kappa\cdot\ell\cdot \log{n} \enspace .
\end{align}
\end{observation}
\begin{proof}[Proof Sketch]
The proof idea is the following. \autoref{eq:bound-on-working-time} follows
from the fact that at every tick the working time and the real time are
simultaneously increased by one, and whenever the working time is set to the
median of the sampled real times, which are also incremented upon each tick,
the property also holds. For the proof of \eqref{eq:time-until-completion} and
\eqref{eq:max-number-of-ticks}, observe that according to
\autoref{alg:asynchronous} a node completes the execution of the algorithm when
$T'_v$ reaches $\kappa \cdot \ell \cdot \log{n}$. The proof of
\eqref{eq:time-until-completion} and \eqref{eq:max-number-of-ticks} follows,
for $\kappa \cdot \ell$ large enough, from an application of Chernoff bounds to
$T_v(\mathfrak{T})$ and union bound over all nodes, where we use
\eqref{eq:bound-on-working-time} to show the second part of
\eqref{eq:max-number-of-ticks}.
\end{proof}

We proceed to show that \emph{most} nodes are \emph{almost synchronous} at
carefully chosen reference points. Intuitively, a huge fraction of nodes has a
number of ticks that is concentrated around the expected value and therefore
most nodes will execute instructions which are \emph{close} together. We
formalize this concept in the following lemma which is based on
\autoref{def:delta-close}. The lemma establishes in its first part that $n
\cdot \left( 1 - \Exp{-\BigTheta{\log{n}/\log^2\log{n}}}\right)$ nodes will be
$(\Delta/6)$-close w.r.t.\ the real time over the course of the algorithm.

In the second statement we consider shorter intervals of the length of a phase and
claim that a much larger number of nodes, to be specific, 
$n \cdot \left( 1 - \Exp{-{9\log{n}/\log\log{n}}}\right)$ nodes, will be selected 
to tick for the same number of times up to an error of $\Delta/16$.

\begin{lemma}\label{lem:how-many-alive}
Let $\Delta \geq c_\Delta \log n/\log \log n$, for some large enough constant $c_\Delta$.
Let $\tau$ be a reference point
with $\tau \leq c\cdot\log{n}$, and let $Y(\tau)$ be the random variable for
the number of nodes which are $(\Delta/16)$-close to $\tau$ w.r.t.~$T_v$. We have
\begin{equation*}
Y(\tau) \geq n \cdot \left( 1 - \Exp{-\BigOmega{\log{n}/\log^2\log{n}}}\right) \enspace .
\end{equation*}
Furthermore, consider an arbitrary interval consisting of $t$ consecutive  ticks.
Fix a subset $Y\subseteq V$ and let $Y'\subset Y$ be the subset of nodes  which 
receive at least $t/n - \Delta/16$ ticks and 
at most $t/n + \Delta/16$ ticks out of the $t$ ticks.
We have
\begin{equation*}
|Y'| \geq |Y| \cdot \left( 1 - \Exp{-10\log{n}/\log\log{n}}\right) -\tilde O(\sqrt{n})\enspace .
\end{equation*}
\end{lemma}

\begin{proof}
Let $\mathcal{E}_v(\tau)$ be the event that a node $v$ is $(\Delta/16)$-close to
$\tau$, that is,
\begin{align*}
\mathcal{E}_v(\tau) & = \big[\tau -\Delta/16 \leq\ T_v(\tau)\ \leq\ \tau +\Delta/16 \big] \enspace .
\intertext{We apply Chernoff bounds to $T_v(t)$ and obtain}
\Probability{\mathcal{E}_v(\tau)} & \geq 1 - \Exp{-\BigOmega{\frac{\log{n}}{\log^2\log{n}}}} \enspace , \label{eq:probability-synchronicity} \numberthis
\end{align*}
Let in the following $Y_v(\tau)$ be an indicator random variable for
a node $v$ and a reference point $\tau$ defined as
\begin{equation*}
Y_v(\tau) = \begin{cases}1, & \text{ if } \mathcal{E}_v(\tau) \enspace , \\ 0, & \text{otherwise.} \end{cases}
\end{equation*}
Summing up over all nodes gives us $Y(\tau) = \sum_{v \in V} Y_v(\tau)$. By
linearity of expectation, we have
$\Expected{Y(\tau)} \geq n\cdot\left(1 - \Exp{-\BigTheta{\log{n}/(\log^2\log{n})}}\right)$.
Note that the random variables $T_v(\tau)$, and therefore also the random
variables $Y_v(\tau)$, are not independent. We thus consider the process of
uncovering $Y_{v}(\tau)$ one node after the other in order to obtain the Doob
martingale of $Y(\tau)$ as follows. We define the sequence $Z_j(\tau)$ as
$Z_j(\tau) = \Expected{Y(\tau) \left| T_j(\tau),\dots, T_1(\tau) \right.}$
with $Z_0(\tau)=\Expected{Y(\tau)}$. We have 
\begin{align*}
\Expected{Z_j(\tau)\left|T_{j-1}(\tau),\dots, T_1(\tau)\right.} &= \Expected{\Expected{Y(\tau) \left| T_j(\tau),\dots, T_1(\tau) \right.} \left|T_{j-1}(\tau),\dots, T_1(\tau)\right.} \\
\intertext{which, applying the tower property, gives us that}
\Expected{Z_j(\tau)\left|T_{j-1}(\tau),\dots, T_1(\tau)\right.} & = \Expected{Y(\tau)\left|T_{j-1}(\tau),\dots, T_1(\tau)\right.}  = Z_{j-1}(\tau) \enspace .
\end{align*}
Therefore $Z_j(\tau)$ is indeed the Doob martingale of $Y(\tau)$.

According to \autoref{obs:max-number-of-ticks} each node ticks at
most $2c \cdot \log{n}$ times, that is, $|T_{j+1}(\tau) -T_j(\tau) |\leq 2 c\cdot \log n$.
This holds \whp in the original process $P$ and with probability $1$ in the
coupled process $P'$. Since at most $2c\cdot\log{n}$ of the random variables
$Y_{j+1}(\tau), \dots, Y_{n}(\tau)$ differ, we have
\begin{align*}
|Z_{j+1}(\tau)- Z_j(\tau)| &= \big|\Expected{Y_n(\tau) + \dots + Y_1(\tau) \left| T_{j+1}(\tau),\dots, T_1(\tau) \right.}\\
& \phantom{{}={}}- \Expected{Y_n(\tau) + \dots + Y_1(\tau) \left| T_j(\tau),\dots, T_1(\tau) \right.}\big|
\leq 2 c \cdot \log{n} \enspace .
\end{align*}
Applying the Azuma-Hoeffding bound to $Y(\tau) = \sum_{v \in V} Y_v(\tau)$
gives us
\begin{equation*}
\Probability{\left|Y(\tau) - \Expected{Y(\tau)}\right| \geq \sqrt{c^3\cdot n\cdot\log^3{n}}} \leq \Exp{-\frac{c^3\cdot n\cdot\log^3{n}}{2\cdot\sum_{j=1}^{n}(2c\cdot\log{n})^2}} \enspace ,
\end{equation*}
which for sufficiently large $c$ yields
$\left|Y(\tau)-\Expected{Y(\tau)}\right| \leq \sqrt{c^3\cdot n\cdot\log^3{n}}$  \whp.
Observe that $\sqrt{c^3\cdot n\cdot\log^3{n}} \leq n\cdot \Exp{-\BigTheta{\log{n}/\log^2\log{n}}}$. 
We finally conclude that, \whp, at least
$n\cdot\left(1-\Exp{-\BigTheta{\log{n}/\log^2\log{n}}}\right)$ nodes are
synchronous up to a deviation of at most $\Delta =
\BigTheta{\log{n}/\log\log{n}}$ ticks from the expected number of ticks at the
given reference point $\tau$.

We now turn to the second part of the statement. Recall that  $\Delta=c_\Delta
\log n/\log \log n$ and $c_\Delta$ is a large enough constant. Observe that, by
definition of our algorithm, $T = 10 \Delta$. The proof of the second part
follows in a similar way as before. We define an analogous event
$\mathcal{E}'_v(\tau_1)$ for node $v$ to hold, then the number of ticks it
receives $t/n \pm \Delta/16$ out of $t$ ticks. We have
\[ \Probability{\mathcal{E}'_v(\tau)}  \geq 1 - \Exp{-{\frac{10\log{n}}{\log\log{n}}}}. \]
Observe that this is bound is much stronger than  \eqref{eq:probability-synchronicity}.
Similarly, as before, $\left|Y'-\Expected{Y'}\right| \leq \sqrt{c^3\cdot n\cdot\log^3{n}}$ \whp. 
Thus,
\[|Y'(\tau_1)| \geq |Y| \cdot \left( 1 - \Exp{-10\log{n}/\log\log{n}}\right) -
\sqrt{c^3\cdot n\cdot\log^3{n}}  \]  yielding the claim.
\end{proof}

\begin{theorem}[Equation 10 from \cite{HR90}] \label{thm:bounds-binomial-distribution}
Let $Y=\sum_{i=1}^m Y_i$ be the sum of $m$ i.i.d.~random variables with
$\Probability{Y_i=1}=p$ and $\Probability{Y_i=0}=1-p$. We have for any
$\alpha \in (0,1)$ that
\begin{equation*}
\Probability{Y \geq \alpha \cdot m } \leq \left(\left(\frac{p}{\alpha} \right)^\alpha \left(\frac{1-p}{1-\alpha} \right)^{1-\alpha}\right)^m \enspace .
\end{equation*}
\end{theorem}

\noindent In the following we show that the median taken will be concentrated around the expected real time.
\begin{lemma}\label{cor:nicemedian}
The median real-time of a uniform sample of $\Omega(\log^2 \log n)$ nodes is
$(\Delta/16)$-close \whp at any reference point $\tau \leq \kappa \cdot \ell \cdot \log n$.
\end{lemma}
\begin{proof}
In this proof we assume for simplicity that the $c'' \log^2 \log n$ sampled
nodes are taken in one single step.  First, we show that the median of the
sampled times is close to the average of all (real) times, \whp. The median
real-time of the sample is no $(\Delta/16)$-close if at least half of the
sample contained nodes which were not $(\Delta/16)$-close. By
\autoref{lem:how-many-alive}, we know that for some constant $c>0$ there are
\whp at most \[ L=n\Exp{-c({\log n/\log^2 \log n}})\] nodes $u$ which are not
$(\Delta/16)$-close w.r.t.~$T_u$ during any point of the  execution of the
algorithm.
  
Let $\mathcal{G}$ be the set of these \emph{bad} nodes. Let $Z$ denote the
number of samples drawn which are bad. Thus, by
\autoref{thm:bounds-binomial-distribution} with parameters $\alpha=1/2$ and
$p=L/n$, we derive 
  \begin{align*}
  	\Probability{Z \geq \alpha c'' \log^2\log n/2 } &\leq \left(\left( 2p \right)^{1/2} \left( 2(1-p) \right)^{1/2}\right)^{c'' \log^2\log n} \leq  2^{c'' \log^2\log n}\cdot (p^{1/2}(1-p))^{c'' \log^2\log n} \\
  	&\leq \sqrt{n}\cdot (L/n)^{c'' \log^2\log n/2} = \sqrt{n} \cdot n^{-\frac{c\cdot c'' \log^2\log n}{2 \log^2\log n }} \leq 1/n^2,
  \end{align*} 
  for large enough $c''$.
\end{proof}

\begin{proof}[Proof of \autoref{prop:bulk}]
\def\tpl{\tp{l}{}}
\def\tpr{\tp{r}{}}
For every phase $s=\BigO{\log \log n}$, let $J_s$ be the set of nodes which are
\begin{enumerate}
\item $(5\Delta/16)$-close w.r.t.\ the working time at any reference point $\tau = s\cdot T $ and
\item $(\Delta/2)$-close w.r.t.\ the working time at any reference point in $[(s-1)\cdot T, s\cdot T]$.
\end{enumerate}
In the following, we show by induction that \whp
\[ |J_s| \geq n \left( 1 - T^2\cdot s\cdot \Exp{-9\log{n}/\log\log{n}}\right) \enspace . \]

For $s=0$ this holds trivially since $|J_0|=n$. Suppose the claims holds for
phase $s$ and consider phase $s+1$. We seek to show that the claim holds in the
interval $[s\cdot T, (s+1)\cdot T]$. Let $\tpl, \tpr$ with $\tpl < \tpr$ be an
arbitrary pair of reference points with $\tpl \geq s\cdot T$ and $\tpr \leq
(s+1) \cdot T$. Let furthermore $J' \subset J_s$ denote the set of nodes which
are selected to tick $ \tpr-\tpl \pm \Delta/16 $ times in any interval
$[\tpl,\tpr]$. By Part 2 of \autoref{lem:how-many-alive}, we have
\begin{equation} \label{eq:size-of-J'}
|J'| \geq |J_s|\left( 1 -\Exp{-9\log{n}/\log\log{n}}\right) \enspace .
\end{equation} 

Let $J'_s$ be the set of nodes which are selected $\tpr-\tpl \pm \Delta/16$
times to tick in \emph{every} interval $[ \tpl, \tpr ]$. Since there are at
most $T^2$ such intervals, we get by \eqref{eq:size-of-J'} that \whp
\[ |J'_s| \geq |J_s|\left( 1 -T^2\cdot\Exp{-9\log{n}/\log\log{n}}\right) \enspace . \]

Let $v$ be an arbitrary but fixed node. Let $\vartheta_v$ be the exact time
step at which $v$ jumps and observe that $\vartheta_v$ is a random variable.
Let furthermore $\tau_v$ denote the first reference point after time step
$\vartheta_v$, that is, $\tau_v = \ceil{\vartheta_v/n}$. Consider the number of
times $v$ is selected to tick in the interval of time steps $[\vartheta_v,
\tau_v\cdot n]$. By a standard balls-into-bins argument \cite{RS98}, we can
argue that \whp
\begin{equation}
| T(\tau_v) - T(\vartheta_v /n)| \leq \Delta/16 \enspace . \label{eq:deviation-before-jump}
\end{equation}

Let $\tau'$ be any reference point in $[\tau_v, (s+1)\cdot T]$.
Since the working time increases afterwards whenever $v$ is selected to tick, we have
\begin{equation}\label{eq:deviation}
T'_v(\tau')-T'_v(\vartheta_v/n)=T_v(\tau')-T_v(\vartheta_v/n) \enspace .
\end{equation}

We now show that every node $v\in J_s'$ jumps exactly once. Recall that \tJUMP
is the instruction at which every node executes the jump step. That is, if any
nodes has a working time of $s\cdot T +\tJUMP$, then that node jumps We claim
that every node $v \in J_s'$ must have jumped prior to $(s+1)\cdot T$, that is,
we have $\tau_v \leq (s+1)\cdot T$. To see this, assume that $v$ did not jump.
By
\eqref{eq:deviation},
 \begin{align*}
T'_v((s+1)\cdot T)&=T_v((s+1)\cdot T)-T_v(s\cdot T)+T'_v(s\cdot T) \\
&\geq (s+1)\cdot T-s\cdot T-\Delta/16 +T'_v(s\cdot T)\\
& \geq (s+1)\cdot T-s\cdot T-\Delta/16 + s \cdot T - 5\Delta/16\\
&> (s+1)\cdot T-\Delta/2 \geq s\cdot T+ \tJUMP \enspace  ,	
 \end{align*}
where the first inequality follows from the definition of $J'_s$ and the second
inequality follows from the induction hypothesis. The the above inequality
implies that $v$ must have executed the jump instruction and thus must have
jumped.

Symmetrically, we claim that
every node $v \in J_s'$ will jump at most once per phase \whp. It suffices to
show that no node of $J_s'$ jumps before reference point $\tau' := \tM2 +
\Delta/2  $, since, informally speaking, at reference point $\tau'  $ all nodes
of $J_s'$ will  have a real time exceeding \tM2 (similarly as before, this can
be shown using the definition of $J_s'$ and the induction hypothesis). Thus, by
\autoref{cor:nicemedian} and the due to the immense size of $J_s'$, node $v$
will set its working time to the median of sampled real times which will be
larger than \tM2. Node $v$ will not execute the jump instruction again in this
phase. To show this claim we need to show that $T'_v(  \tau' ) <s\cdot T +
\tJUMP$, which is true since
\eqref{eq:deviation},
 \begin{align*}
T'_v(  \tau' )&=T_v(  \tau'  )
-T_v(   s\cdot T)+T'_v( s\cdot T) \\
&\leq \tM2 +\Delta/2 +\Delta/16 +T'_v(s\cdot T)\\
& \leq \tM2 +\Delta/2+\Delta/16 + s \cdot T - 5\Delta/16\\
&\leq  (s+1)\cdot T-\Delta/2 =  s\cdot T + \tJUMP \enspace  ,	
 \end{align*}
where the first inequality follows from the definition of $J'_s$ and the second
inequality follows from the induction hypothesis. Thus, $v$ jumped at most
once. We therefore conclude that every node $v \in J_s'$ jumps exactly once.

We will now argue the following. For every $v \in J_s'$ chooses \whp
\begin{equation}\label{eq:deviation-after-jump}
|T'_v(\vartheta_v/n) - \vartheta_v /n| \leq  2\Delta/16 + 1 \enspace .
\end{equation}
To see this, first observe that, by \autoref{cor:nicemedian}, the median taken
from $\log^3\log n$ samples of the real time is  $(\Delta/6)$-close. Second, we
need to account for the fact that median is not taken directly, but rather over
time. If all samples were taken directly before jumping, then the median would
indeed be $(\Delta/6)$-close. However, since $v\in J_s'$, it holds that the
value of any sample is $(\Delta/6)$-close w.r.t. the value it would have if it
were sampled directly before $v$ jumps. Accounting for all errors, using
triangle inequality and that  $\tau_u=\ceil{\vartheta_v/n}$,
\eqref{eq:deviation-after-jump} follows.

We proceed by showing that after $v \in J'_s$ jumps its working-time
well-concentrated, that is,
\begin{equation}\label{eq:working-time-after-jump}
|T'_v(\tau')-\tau'|     \leq 5\Delta/16 \enspace ,
\end{equation}
for any reference point $\tau'$   in $[\tau_v, (s+1)\cdot T]$.
We have 
\begin{align*}
T'_v(\tau')&\stackrel{\eqref{eq:deviation}}{=}T'_v(\vartheta_v/n) + T_v(\tau')-T_v(\vartheta_v/n)\\
&\stackrel{\eqref{eq:deviation-after-jump}}{\leq} \vartheta_v /n + 2\Delta/16 + 1 +     T_v(\tau')-T_v(\vartheta_v/n) \\
&\stackrel{\eqref{eq:deviation-before-jump}}{\leq}   \vartheta_v /n + 2\Delta/16 + 1 +     T_v(\tau')-T_v(\tau_v)+\Delta/16 \\
&\stackrel{\text{def.\ }J_s'}{\leq}   \vartheta_v /n + 2\Delta/16 + 1 +     (  (\tau'-\tau_v)+ \Delta/16 )+ \Delta/16\\
&\stackrel{\text{def.\ }\tau_v}{\leq}   \tau_v + 1 + 2\Delta/16 + 1 +     (  (\tau'-\tau_v)+ \Delta/16 )+ \Delta/16\\
&\leq \tau'+ 5\Delta/16,    	
\end{align*}
Symmetrically, we have
\begin{align*}
T'_v(\tau')&\stackrel{\eqref{eq:deviation}}{=}T'_v(\vartheta_v/n) + T_v(\tau')-T_v(\vartheta_v/n)\\
&\stackrel{\eqref{eq:deviation-after-jump}}{\geq} \vartheta_v /n - 2\Delta/16 - 1 +     T_v(\tau')-T_v(\vartheta_v/n) \\
&\stackrel{\eqref{eq:deviation-before-jump}}{\geq}   \vartheta_v /n - 2\Delta/16 - 1 +     T_v(\tau')-T_v(\tau_v)-\Delta/16 \\
&\stackrel{\text{def.\ }J_s'}{\geq}   \vartheta_v /n - 2\Delta/16 - 1 +     (  (\tau'-\tau_v)+ \Delta/16 )- \Delta/16\\
&\stackrel{\text{def.\ }\tau_v}{\geq}   \tau_v - 1 - 2\Delta/16 - 1 +     (  (\tau'-\tau_v)+ \Delta/16 )- \Delta/16\\
&\geq \tau'- 5\Delta/16 \enspace ,    	
\end{align*}
This shows \eqref{eq:working-time-after-jump}. 
Define $J_{s+1}=J'_s$.
This shows that $v\in J_{s+1}$ is $(5\Delta/16)$-close at $(s+1)\cdot T$.
Furthermore,  at reference point $s\cdot T$, $v$ was, by induction hypothesis, $(5\Delta/16)$-close and, since $J_{s+1}=J_s'$, at every reference point
$\tau$ before $u$ jumped we can derive $|T'_v(\tau)- \tau| \leq  5\Delta/16 + \Delta/16 \leq \Delta/2 $.
Furthermore, \eqref{eq:working-time-after-jump} implies that  $v$ was also $(\Delta/2)$-close after jumping and thus $v$ was $\Delta/2$ at each reference point in $[s\cdot T, (s+1)\cdot T]$.

We now show that $|J_{s+1}|$ is large enough. Using the induction hypothesis, we have
\begin{align*}
	|J_{s+1}| = |J'_s| &\geq |J_s| \left( 1 - T^2\cdot \Exp{-9\log{n}/\log\log{n}}\right) \\
	&\geq  n  \left( 1 - sT^2\cdot \Exp{-9\log{n}/\log\log{n}}\right) \left( 1 - T^2\cdot \Exp{-9\log{n}/\log\log{n}}\right)\\
	&\geq n \left( 1 - (s+1)T^2\cdot \Exp{-9\log{n}/\log\log{n}}\right) \enspace .
\end{align*}
This finishes the induction step. Finally, observe that for any  $s=\BigO{\log \log n}$ we have  
\[ n \cdot \left( 1 - s\cdot T^2\cdot \Exp{-9\log{n}/\log\log{n}}\right)\geq n \left( 1 - \Exp{-8\log{n}/\log\log{n}}\right) \enspace . \qedhere \]
\end{proof}

\subsection{Analysis of the  \TC sub-phase: Proof of \autoref{prop:two-choices}}\label{sec:proofas1}

\begin{proof}[Proof of \autoref{prop:two-choices}]
Recall that \bulk is the set of nodes
$v$ that are $(\Delta/2)$-close w.r.t.\ $T'(v)$ throughout the entire process. 
By \autoref{prop:bulk}, 
$|\bulk| \geq n - \mathcal{E}$, with $\mathcal{E} \leq n\cdot\Exp{-8\log{n}/\log\log{n}})=n^{1-{8}/{\log\log n}}.$
When a node of \bulk samples two nodes, then by definition the working time of all nodes of \bulk is larger than 
$\tau_0$ and smaller than \tSET. 
Let $u$ be a node of \bulk. Then, $u$ samples at two nodes (that is, when its working time is \tTC{}), then its probability of sampling two nodes of color \C{j} with probability at least 
$( \hat c_j(\tau_0)/n)^2$ and at most $((\hat c_j(\tau_0)+\mathcal{E})/n)^2$. 

By Chernoff bounds, \[ x_1(\tBP{1}) \geq  |\bulk| \cdot   ( \hat c_j(\tau_0)/n)^2 - \sqrt{n }\log n \geq \frac{ \hat c_j(\tau_0)^2}{n}\left(1 - \LittleO{1}\right) \enspace ,\]
where we used the fact that all nodes of \bulk must have executed the instruction at \tSET at reference point \tBP{1}.

We now distinguish between two cases. If $\hat c_j(\tau_0) \leq n^{1-{7}/{\log\log n}}$ we have, 
$\hat c_j(\tau_0) + \mathcal{E} = \BigO{ n^{1-{7}/{\log\log n}} }$.
Thus, by Chernoff bounds, \whp
 \[x_j(\tBP{1}) \leq n\cdot ((\hat c_j(\tau_0) + \mathcal{E})/n)^2 + \sqrt{n }\log n =  \BigO{n^{1-{14}/{\log\log n}}} \enspace .\]

Otherwise,  $\hat c_j(\tBP{1}) > n^{1-{7}/{\log\log n}}$  and we have $\hat c_j(\tau_0) + \mathcal{E}=\hat c_j(\tau_0)(1 + \LittleO{1})$.
Thus, by Chernoff bounds, we obtain \whp that \[ x_j(\tBP{1})  \leq  n  \cdot ((\hat c_j(\tau_0) + \mathcal{E})/n)^2 
 = \hat c_j(\tau_0)^2/n \cdot (1 + \LittleO{1}) \enspace . \]
This finishes the proof. 
\end{proof}

\subsection{Analysis of the \BP Sub-Phase: Proof of \autoref{prop:bit-propagation}}\label{sec:proofas2}

We now focus on the analysis of the \BP sub-phase. Similar to the analysis of
the synchronous case, we first analyze the number of bits which are set during
the \BP sub-phase without taking their color into consideration. The following
lemma is based on the observation that the \BP can be modeled by a simple
asynchronous randomized-gossip-based information dissemination process.

\begin{lemma}\label{lem:number-of-bits-after-BP}
Consider an arbitrary but fixed phase and let $x(\tau)$ be the number of nodes
in \bulk which have a bit set at reference point $\tau$ in that phase. Assume
that $|\bulk| \geq n \cdot \left( 1 - \Exp{-8\log{n}/\log\log{n}}\right)$ and
that $x(\tp{BP}{1}) \geq n/(2k)$. Then we have $x(\tp{BP}{2}) = |\bulk|$ \whp.
\end{lemma}

\begin{proof} 
We split the proof into three parts, in each of which we will rely on the fact
that at each reference point the nodes of \bulk are $(\Delta/2)$-close. We
argue that \whp (i) $x(\tau_2') \geq n/2$, (ii) $x(\tau_4) \geq
|S|\cdot\left(1-n^{-2/\log\log{n}}\right)$, and (iii) $x(\tBP2) = |\bulk|$.

\paragraph{Part (i).}
To show the first part, we first consider a sequence of $\Delta$ periods from
$\tau_1'$ to $\tau_2'$. Recall that each period consists of $n$ consecutive
time steps. We will show by induction over $i \in [\tau_1', \tau_2')$ that
\begin{equation*}
x(i) \geq \min\left\{\frac{n}{2},~\frac{n}{2k}\cdot\left(1+\frac{1}{5}\right)^i\right\} \enspace .
\end{equation*}

Let $i$ be an arbitrary but fixed period in $ [\tau_1', \tau_2')$ and assume
that $x(i-1) < n/2$. Note that by definition of \bulk at any reference point
$\tau \in [\tau_1',\tau_2']$ all nodes of \bulk are in $[\tau_1,\tau_3]$. Let
$H(i) \subseteq \bulk$ be the set of nodes in \bulk which did not have their
bit set after period $i-1$. By assumption, $|H(i)|\geq |\bulk| - n/2 = n/2
\cdot\left(1-\LittleO{1}\right)$. Let furthermore $A(i)$ be the set of
\emph{active} nodes which tick in period $i$ at least once. By a standard
balls-into-bins arguments \cite{RS98}, we have that $|A(i)|$ has size at least
$n/2$ \whp. Observe that each node is equally likely to tick, independently of whether the bit is set or not.
Therefore, $A(i)$ and $H(i)$ are independent, and any node in $H(i)$ ticks at least once with
probability at least $n/2$, independently. Hence, $|A(i)\cap H(i)| \geq n/4
\cdot\left(1-\LittleO{1}\right) $ \whp, where the concentration follows from
Chernoff bounds.

For a node $v \in A(i)\cap H(i)$ in period $i$, we define $X_v$ to be the
indicator random variable for the event that $v$ sets the bit. Note that all
$X_v$ are independent and $\Probability{X_v = 1} \geq x(i-1)/n$. Let $X=\sum
X_i$. By Chernoff bounds, $X \geq |A(i) \cap H(i)| \cdot x(i-1)/n \cdot
\left(1-\LittleO{1}\right) \geq x(i-1)/5$ \whp. We therefore get that \whp 
\begin{equation*}
x(i) \geq x(i-1)+ X \geq x(i-1) \left(1+\frac{1}{5} \right) \stackrel{\text{IH}}{\geq} \frac{ n}{2k} \cdot \left(1+\frac{1}{5} \right)^i \enspace ,
\end{equation*}
which completes the induction. We now obtain, using ${\tau_2'-\tau_1'}\geq
4\log k$, that
\begin{equation*}
x(\tau_2')\geq \frac{n}{2k}\left(1+\frac{1}{5}\right)^{\tau_2'-\tau_1'} \geq \frac{n}{2k} \cdot k = n/2 \enspace .
\end{equation*}
This completes the proof of Part (i).

\paragraph{Part (ii).}
Let $H(\tau_2') \subseteq \bulk$ be the set of nodes in \bulk which do not have
a bit set at reference point $\tau_2'$. We consider an arbitrary but fixed node
$v \in H(\tau_2')$ at reference point $\tau_4$. Since $v$ is in \bulk and thus
$(\Delta/2)$-close at both, $\tau_2'$ and $\tau_4$, we observe that it ticked
at least $ \tau_4-\tau_2'-2\cdot\Delta/2 = \Delta/2$ times between time steps
$\tau_2' \cdot n$ and $\tau_4 \cdot n$ corresponding to these reference points.
The probability that the node $v$ never sampled a node with the bit set is thus
at most $2^{-\Delta/2}$. Hence, by using independence and Chernoff bounds, the
number of nodes remaining in $H(\tau_4)$ is, for $\Delta$ large enough, at most
$|\bulk| \cdot n^{-2/\log \log n}$ \whp.

\paragraph{Part (iii).}
As before, let $H(\tau_4) \subseteq \bulk$ be the set of nodes in \bulk which
do not have a bit set at reference point $\tau_4$. We again consider an
arbitrary but fixed node $v \in H(\tau_4)$. Since $v$ is in \bulk and thus
$(\Delta/2)$-close at both, $\tau_4$ and $\tBP2$, we observe that it performed
at least $ \tau_5-\tau_4' = \Delta/2$ \BP ticks. The probability that $v$ samples
in one of these ticks a node in \bulk without the bit set or that $v$ samples a node not in \bulk is at most
$n^{-2/\log\log{n}} + n^{-8/\log\log{n}} \leq n^{-1/\log\log{n}}$. Therefore, the probability
that this node never obtains the bit is at most
$\left(n^{-1/\log\log n}\right)^{ \Delta/2 } \leq n^{-\LittleOmega{1}}$.
From union bound we derive that all nodes in \bulk therefore have the bit set
at reference point $\tBP2$.
\end{proof}

In the following we analyze the individual colors during the \BP sub-phase. Our
main observation is that the \BP process can be modeled by so-called Pólya urns
\cite{JK77}. In this model, we are given an urn containing marbles of two
colors, black and white. In every step, one marble is drawn uniformly at random
from the urn. Its color is observed, the marble is returned to the urn and one
more marble of the same color is added. For any color, the ratio of marbles
with that given color over the total number of marbles is a martingale. We will
use this urn process to model the \BP sub-phase, which then can be
analyzed by means of martingale techniques. Formally, the Pólya urn process is
defined as follows.

\newcommand{\polya}[2]{\ensuremath{\operatorname{P\acute{o}lya}\left(#1,#2\right)}\xspace}
\begin{definition}[Pólya Urn Process]
Let \polya{\alpha_1}{\alpha_2} with $\alpha_1, \alpha_2 \in \mathbb{Z}^+_0$ be the following urn process. At the
beginning there are $\alpha_1$ black marbles and $\alpha_2$ white marbles in the urn. The process runs in multiple steps where
$\alpha_1(i)$ and $\alpha_2(i)$ denote the number of black and white marbles in the urn,
respectively, for every time step $i$. In every time step $i$, a black
marble is added with probability  $\alpha_1(i)/(\alpha_1(i)+\alpha_2(i))$, and with remaining
probability $\alpha_2(i)/(\alpha_1(i)+\alpha_2(i))$ a white marble is added.
\end{definition}

We now use this urn model to show our main result for the \BP sub-phase,
\autoref{prop:bit-propagation}. We start by performing a worst-case analysis
for color \A in order to give a lower bound on the number of nodes of color \A
after the \BP sub-phase. Similarly, we will upper bound any \emph{large} color
\C{j}. Then we will show that after each phase the gap between \A and \C{j}
grows quadratically. We will use bounds resulting from
\autoref{prop:two-choices} for the numbers of nodes with bits and their color distribution among \bulk. For
the worst case analysis, we will assume that any node which is not in \bulk has
color \C{j} and its bit set. We now give the formal proof.

\begin{proof}[Proof of \autoref{prop:bit-propagation}]
We consider an arbitrary but fixed \BP sub-phase which we model by \polya{\alpha_1}{\alpha_2}
as follows. Initially, we place for each node in \bulk of color \A which has its bit set
at reference point $\tp{BP}{1}$ a black marble in the urn, that is, $\alpha_1 =
x_1(\tp{BP}{1})$. Additionally, we add for each node in \bulk which has its bit set for
any color $\C{j} \neq \A$ a white marble in the urn. Finally, in order to
perform a worst-case analysis, we add a white marble for any node which is not
in \bulk, that is, we add an additional number of $|V \setminus \bulk|$ white
marbles. We therefore have $\alpha_1 + \alpha_2 = x(\tp{BP}{1}) + |V \setminus \bulk|$. We
now consider only those time steps of the \BP sub-phase, where a node in \bulk
without bit samples another node with bit. We couple these very steps with the
Pólya urn process, where we assume that a marble is added based on the adopted
color in the \BP process, that is, if a node newly adopts a bit for
color \A, we add a black marble, and if otherwise a node adopts a bit for color
$\C{j} \neq \A$, we add a white marble. For the worst-case analysis we assume
in the \BP process that all nodes in $V \setminus \bulk$ have a bit set for a
color $\C{j} \neq \A$ throughout the entire process. This corresponds to the
additional $|V \setminus \bulk|$ white marbles initially added to the urn.

As before, we will use the notation that $x(\tau)$ denotes the number of nodes
in \bulk which have a bit set at reference point $\tau$ and $x_j(\tau)$ denotes
the number of nodes in \bulk of color \C{j} which have a bit set at reference
point $\tau$. Let $M$ be a lower bound on $x(\tp{BP}{1})$, the number of bits
set at the beginning of the \BP sub-phase, and recall that according to the
proof of \autoref{prop:two-choices} we have \whp
\begin{equation} M \geq n/(2k) \enspace . \label{eq:lowerbound-bits} \end{equation}

We now consider the Pólya urn process. Let $F(i)$ be the fraction of black marbles in step $i$ of the Pólya urn
process. As mentioned before, this fraction of black marbles in the Pólya urn process
is a martingale. Observe furthermore that $|F(i) - F(i-1)| \leq 1/M$
throughout the entire urn process. Let $\mathcal I$ be the last step of the Pólya urn
process and observe that $\mathcal I \leq n$. Applying Azuma's inequality to $F(i)$
for any $i \leq \mathcal I$ gives us
\begin{align*}
\Probability{|F(i) - F(1) | \geq \delta} &\leq 2\cdot \Exp{-\frac{\delta^2}{ 2\cdot\sum_{j=1}^{i}1/M^2}} \\
&\leq 2\cdot \Exp{-\frac{\delta^2\cdot M^2}{ 2\cdot i}} \enspace .
\intertext{We set $\delta = 4\cdot k\cdot\sqrt{\log{n}/n}$ and obtain using \eqref{eq:lowerbound-bits}}
\Probability{|F(i) - F(1)| \geq 4 \cdot k\cdot\sqrt{\log{n}/n}} &\leq 2\cdot \Exp{-\frac{ 2\cdot k^2 \cdot M^2 \cdot \log{n}}{ n\cdot i }} \\
&\leq 2\cdot \Exp{-2\cdot\log{n}} \enspace , \label{eq:fraction-deviation-bound} \numberthis
\end{align*}
where we used that $x(\tBP1) \geq n/(2k)$ \whp.

From the calculation above we see that \whp the fraction of black marbles in
the urn remains concentrated around the initial value. To derive a lower bound
on the absolute number of black marbles at the end of the process we first
bound $F(1)$. By \autoref{prop:bulk}, we have $|V\setminus\mathcal{S}| \leq
n^{1-8/\log\log n}$ and thus
\begin{align}
F(1) &\geq \frac{ x_1(\tp{BP}{1}) }{ x(\tp{BP}{1}) + |V \setminus \bulk| } 
\geq  \frac{ x_1(\tp{BP}{1}) }{ x(\tp{BP}{1}) + n^{1-8/\log\log n} } 
= \frac{ x_1(\tp{BP}{1}) }{ x(\tp{BP}{1}) } \cdot \left(1 - \LittleO{1}\right)
\intertext{Using \eqref{eq:fraction-deviation-bound}, we get for the end
of the \BP sub-phase that at reference point \tp{BP}{2} \whp}
F(\mathcal I) & \geq F(1) -  4 \cdot k\cdot\sqrt{\log{n}/n} = \frac{x_1(\tp{BP}{1})}{x(\tBP1)} \cdot (1 - \LittleO{1}) - 4\cdot n^{1/\log\log n} \sqrt{\log n/ n} \notag\\
&=\frac{x_1(\tp{BP}{1})}{x(\tBP1)} \cdot (1 - \LittleO{1}) 
  \enspace ,\notag
\end{align}
where we used that $x_1(\tBP1) \geq n/(2k^2) \geq n^{1-3/\log\log n}$ \whp and $x(\tBP) \leq n$.
Hence,
\begin{align}\label{eq:bounds-for-a}
	x_1(\tBP{2})  \geq x(\tBP{2}) \frac{x_1(\tp{BP}{1})}{x(\tBP1)} \cdot (1 - \LittleO{1}) \enspace .
\end{align}  
It remains to establish an upper bound on $x_j(\tBP2)$ for every other large
color $\C{j} \neq \A$. We will use a symmetric argument. Let $\C{j} \neq A$ be
an arbitrary but fixed color and let $F'(i)$ be the fraction of black marbles
in another Pólya urn process which we use to bound the size of color \C{j}. As
before, we use the black marbles to represent \C{j}, the color under
investigation, and the white marbles to represent all other colors $\C{i} \neq
\C{j}$. For the worst case analysis, we again assume that all nodes of $V
\setminus \bulk$ have their bit set for color \C{j}. We apply a similar
computation as before and observe, now for color \C{j}, that
\begin{align*}
F'(1) &\leq \frac{ x_j(\tp{BP}{1}) + |V \setminus \bulk| }{ x(\tp{BP}{1}) + |V \setminus \bulk| } \leq
\frac{ x_j(\tp{BP}{1}) + |V \setminus \bulk| }{ x(\tp{BP}{1})  }  \\
&\leq
\frac{ x_j(\tp{BP}{1}) + n^{1-8/\log\log n} }{ x(\tp{BP}{1})  }   \\
&\leq
\frac{ x_j(\tp{BP}{1})  }{ x(\tp{BP}{1})  }   +
\frac{  n^{1-8/\log\log n} }{ n^{1-3/\log\log n}  }\\
&\leq \frac{ x_j(\tp{BP}{1}) }{ x(\tp{BP}{1}) }  + n^{-5/\log\log n}\enspace .
\intertext{Again using \eqref{eq:fraction-deviation-bound}, we get \whp}
F'(\mathcal I) & \leq F'(1) + 4 \cdot k\cdot\sqrt{\log{n}/n} =
 \frac{x_1(\tp{BP}{1})}{x(\tBP1)}  + n^{-5/\log\log n}\enspace  + n^{-1/3}
\leq \frac{x_1(\tp{BP}{1})}{x(\tBP1)}  + 2n^{-5/\log\log n} \enspace .
\intertext{Thus, using that $x(\tBP2)/x(\tBP2)\leq 2k$ \whp we get}
x_j(\tp{BP}{2}) &\leq x(\tp{BP}{2}) \cdot \frac{x_j(\tp{BP}{1})}{x(\tBP1)}  + 2n^{-1/\log\log n} \cdot 2n^{-5/\log\log n}   = x(\tp{BP}{2}) \cdot \frac{x_j(\tp{BP}{1})}{x(\tBP1)}  + 4n^{-5/\log\log n}  \enspace .
\intertext{Furthermore, from the calculation above and \eqref{eq:bounds-for-a} we obtain for all \C{j} that \whp}
x_j(\tBP{2}) &= x_j(\tp{BP}{1}) \cdot \frac{ x(\tp{BP}{2}) }{ x(\tp{BP}{1}) } \cdot \left(1 \pm \LittleO{1}\right)  + \BigO{n^{-{5}/{\log\log n}}} \enspace . 
\intertext{By \autoref{prop:two-choices}, we have that \whp}
x_j(\tp{BP}{1}) &= \frac{\hat c_j(\tSET)^2}{n}\left(1 \pm \LittleO{1} \right) + \BigO{n^{1-{5}/{\log\log n}}} \enspace .
\intertext{Moreover, by \autoref{lem:number-of-bits-after-BP} and \autoref{def:delta-close}, we have}
x(\tBP2) &\in[ n\cdot(1-\LittleO{1}), n] \enspace .
\intertext{Putting everything together, we derive that \whp }
x_j(\tp{BP}{2}) &= \frac{\hat c_j(\tau_0)^2}{x(\tp{BP}{1})}\left(1\pm \LittleO{1}\right) + \BigO{n^{1-{4}/{\log\log n}}} \enspace . \qedhere
\end{align*}
\end{proof}

\subsection{The Endgame: Taking $a$ from $(1-\teps{Part1})\cdot n$ to $n$} \label{sec:endgame}
\begin{float}[t!]
\begin{minipage}[b]{\textwidth-194.19pt-2em}
\begin{algorithm}[H]
\SetAlgoVlined
\SetKw{KwLet}{let}
\SetKwProg{Algorithm}{Algorithm}{ (Part 2)}{end}
\SetKwFunction{Asynchronous}{asynchronous}
\SetKwFunction{xSamples}{samples}
\SetKwFunction{xMedian}{median}
\SetKwFunction{xColor}{color}
\SetKwFunction{xTemp}{intermediate}
\SetKwFunction{xBit}{bit}
\SetKwFunction{xWorkingTime}{workingtime}
\SetKwFunction{xRealTime}{realtime}
\def\Samples#1{\xSamples{\ensuremath{#1}}}
\def\Median#1{\xMedian{\ensuremath{#1}}}
\def\Color#1{\xColor{\ensuremath{#1}}}
\def\Temp#1{\xTemp{\ensuremath{#1}}}
\def\Bit#1{\xBit{\ensuremath{#1}}}
\def\WorkingTime#1{\xWorkingTime{\ensuremath{#1}}}
\def\RealTime#1{\xRealTime{\ensuremath{#1}}}
\SetKw{KwLet}{let}
\Algorithm{\Asynchronous{node $v$}}{
\If{$\tMETA1 \leq \WorkingTime{v} \leq \tMETA4$}{
\KwLet $u_1, u_2 \in N(v)$ uniformly at random\;
\If{$\Color{u_1} = \Color{u_2}$}{
$\Color{v} \gets \Color{u_1}$\;
}

}

$\WorkingTime{v} \gets \WorkingTime{v} + 1$\;
}
\caption{Part 2 of the asynchronous protocol to solve plurality consensus. At
ticks in $[\tMETA0, \tMETA1]$, the nodes do not perform any action.}
\label{alg:endgame}
\end{algorithm}
\end{minipage}\hfill
\begin{minipage}{194.19pt}
\begin{figure}[H]
\includegraphics{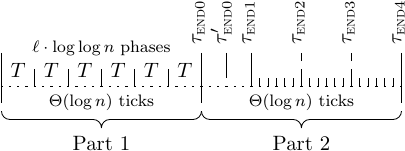}\medskip\stepcounter{figure}

\footnotesize\textbf{Figure 3:} graphical representation of the asynchronous protocol, showing Part 1 (\autoref{alg:asynchronous}) and Part 2 (\autoref{alg:endgame})
\end{figure}
\end{minipage}
\end{float}

In this section we analyze Part 2 of the asynchronous algorithm defined in
\autoref{alg:endgame}. As we will argue in the proof of \autoref{thm:async}, we
assume at for Part 2 that at $\tMETA1$ we have \whp $a = c_1 \geq
(1-\teps{Part1})\cdot n$, where \teps{Part1} is a small constant. Observe that
Part 2 is executed after Part 1 defined in \autoref{alg:asynchronous}.
Therefore, $\tMETA0 = \kappa\cdot\ell\cdot\log{n}$. We define the following
reference points for Part 2.
\begin{align*}
\tMETA0'&= \ifrac{3}{2} \cdot \tMETA0 &
\tMETA1&= 2 \cdot \tMETA0 &
\tMETA2&= 3 \cdot \tMETA0 &
\tMETA3&= 4 \cdot \tMETA0 &
\tMETA4&= 5 \cdot \tMETA0
\end{align*}

Observe that according to the definition of Part 2 given in
\autoref{alg:endgame} we only consider the working time (and not the real time).
As \autoref{obs:max-number-of-ticks} Part 1 suggests, the working times of the nodes 
are sandwiched by the real time of the nodes and thus if we bound the real times
of nodes, we get bounds on the working times as well.

 From
\autoref{obs:max-number-of-ticks} we obtain that all nodes have finished Part 1
at time step $\mathfrak{T}$ after at most $\mathfrak{T} \leq
\ifrac{3}{2}\cdot\kappa\cdot\ell\cdot\log{n} = \tMETA0'$ ticks w.r.t.\ the working time.
Furthermore, also due to \autoref{obs:max-number-of-ticks}, we have that no
node has yet reached \tMETA1 w.r.t.\ the working time at time step $\mathfrak{T}$. Therefore, we conclude
that all nodes have completed Part 1 before any node starts the two choices
process of Part 2 at reference point $\tMETA1$. More precisely, all nodes are \whp in $[\tMETA0, \tMETA1]$
before the first node passes \tMETA1. 

Since the real times are sandwiched, we get from Chernoff bounds that when the first node reaches \tMETA2, all nodes are \whp in
$[\tMETA1, \tMETA2]$ w.r.t.\ the real time. 
We assume that nodes which are in [\tMETA0, \tMETA4] respond, when
queried, with the color they last set, possibly in Part 1 of the algorithm.

The remainder of this section is structured as follows. In
\autoref{lem:neverback} we give a lower bound on the size of $\A$ throughout
the execution of \autoref{alg:two-choices}. This lower bound on $\A$ allows us
to show that the number of nodes having any other color $\C{j} \neq \A$
decreases quickly in expectation. This expected drop lets us apply a standard
drift theorem, \autoref{thm:drift-theorem}, to obtain a bound on the required
time until $\A$ prevails and all other colors vanish. Finally, this will allow
us to show that \whp all nodes have set their color to \A by the end of the the
asynchronous algorithm at \tMETA4. 

For the next two lemmas, we will use the following notation. Consider an
arbitrary but fixed time step $t$. Let $a_t$ and $b_t$ be the number of nodes
of color \A and \B at time step $t$, respectively.

\begin{lemma}\label{lem:neverback}
Assume that all nodes have a working time in $[\tMETA0,\tMETA4]$ during the
time steps in $[n\cdot\tMETA0', n\cdot\tMETA3]$. Assume furthermore
that at time step $t = n\cdot\tMETA0'$ we have $a_t \geq 19n/20$.
Then for any later tick $t'$ in $[n\cdot \tMETA0', n\cdot\tMETA4]$
we have $a_{t'} \geq 4n/5$, \whp.
\end{lemma}
\begin{proof}
To show the claim, we split Part 2 of the asynchronous algorithm
into phases of $n/100$ consecutive time steps each. Based on these phases, we show the claim by an induction
over every phase $i \in [100\cdot\tMETA0', 100\cdot\tMETA4]$.
By induction, we will show that we have \whp at time step $t_i = i \cdot 100 \cdot n$ \[
a_{t_i} \geq 17 n / 20 - i \cdot \sqrt{n} \cdot \log n \enspace .\]

Let now $i$ be an arbitrary but fixed phase. We distinguish two cases.
\paragraph{Case 1: $a_{t_{i}} \geq 18n/20$.} In this case the induction step holds trivially, since in the
worst case $a_{t_{i+1}} \geq a_{t_{i}} - \left( t_{i+1} - t_{i} \right) = 18n/20 - n/100 > 17n/20$.
\paragraph{Case 2: $a_{t_{i}} \leq 18n/20$.}
Observe that we have, by induction hypothesis, that
for every $t \in [t_i, t_{i+1}]$ that
$a_{t}\geq 17 n/20 - i\cdot \sqrt{n} \cdot \log n - n/100 \geq 16.5n/20$.
Furthermore, by assumption of the lemma we have $a_t \geq 19n/20$
at time step $t = n\cdot\tMETA0'$. We conclude that there are at
least $n/20$ nodes that have already passed \tMETA2 and changed their color
away from \A. However, by assumption of the lemma, these nodes have not yet
passed \tMETA4. These nodes can thus switch to \A if they are selected to tick
and choose two nodes of color \A.
 
We define the random variable $X_t$ as $1$ when a node of color $\C{j} \neq \A$
is selected to tick and changes its color to \A and as $-1$ if a node of color
\A is selected to tick and changes its color to any other color $\C{j} \neq
\A$. If neither of these cases apply, we define $X_t$ to be zero. 
Observe, that the probability for $X_t$ to be negative is maximized when
$b_t=n-a_t$. Therefore, we have
\begin{align*}
	X_t = 
	\begin{cases}
		1 & \text{with probability at least } 1/20 \cdot (16.5n/20)^2/n^2=272.25/20^3 \\
		-1 & \text{with probability at most } 19/20 \cdot (3.5n/20)^2/n^2=232.75/20^3 \\
		0 & \text{otherwise.}
	\end{cases}
\end{align*}
We now define $Y_t$ as $Y_t = \sum_{k\leq t} X_k$ and show that $Y_t$ is a
sub-martingale.
\begin{align*}
\Expected{Y_t\left|Y_{t-1},\dots, Y_1\right.} &= Y_{t-1} +\Expected{ X_t\left|Y_{t-1},\dots, Y_1\right.} \\
& \geq Y_{t-1} - 19/20 \cdot (3.5n/20)^2/n^2 +1/20 \cdot (16.5n/20)^2/n^2 \\
& \geq Y_{t-1} \enspace .
\intertext{Since $|Y_t-Y_{t-1}|\leq 1$, applying the Azuma-Hoeffding bound to $Y_t$ gives us}
\Probability{Y_{t_{i+1}} -Y_{t_i}  \geq - \sqrt{n} \cdot \log n} &\leq \Exp{-\frac{n \cdot \log^2 n}{ 2\cdot n / 100} } \enspace ,
\end{align*}
which yields that the induction steps hold \whp. This completes the proof.
\end{proof}

The following is a version of the multiplicative drift theorem which we will
use in \autoref{lem:endgame} to derive a bound on the number of required
periods until all nodes agree on one opinion.

\begin{theorem}[{\cite[Theorem 5]{LS16}}]\label{thm:drift-theorem}
Let $(X_t)_{t \in\mathbb{N}_0}$ be a Markov chain with state space $\mathcal{S}
\subseteq \{0\} \cup [1,\infty)$ and with $X_0=n$. Let $T$ be the random
variable that denotes the earliest point in time $t \geq 0$ such that $X_t =
0$. Assume that there is $\delta>0$ such that for all $x \in \mathcal{S}$
\begin{equation*}
\Expected{X_{t+1}\mid X_t=x} \leq (1-\delta)x \enspace .
\end{equation*}
Then
\begin{equation*}
\Probability{T > \left\lceil \frac{\log n +k}{|\log(1-\delta)|} \right\rceil}\leq e^{-k} \enspace .
\end{equation*}
\end{theorem}

\begin{lemma}\label{lem:endgame}
Assume that all nodes have a working time in $[\tMETA1,\tMETA4]$ during the
time steps in $[n\cdot\tMETA2, n\cdot\tMETA3]$. Furthermore assume that $a_{t}
\geq 4n/5$ for any time step $t\in [n\cdot\tMETA2,n\cdot\tMETA3]$. Then at
reference point \tMETA3 all nodes have opinion $\A$ \whp, that is,
$a_{\tMETA3}=n$.
\end{lemma}
\begin{proof}
W.l.o.g.\ let $b_t=n-a_t$. We have
\begin{align*}
\Expected{b_{t+1}-b_{t}| \mathcal{F}_t } &= (+1) \frac{a_t}{n} \cdot \frac{b_t^2}{n^2} + (-1) \frac{b_t}{n} \frac{a_t^2}{n^2} \\
& = \frac{a_t \cdot b_t(b_t - a_t)}{n^3} \leq \frac{a_t \cdot b_t\cdot (-3/5)n}{n^3} \leq   -\frac{\ifrac{4}{5}n \cdot b_t\cdot \ifrac{3}{5} n}{n^3}  \\
& = -\frac{12 \cdot b_t}{25 n} \enspace .
\end{align*}
Let $\delta = 12/(25n)$ and define $\Phi(x_t)=b_t$. Note that
$\Phi(x_{max}) \leq n$ and at any time step $t$ we have
$\Expected{\Phi(x_{t+1})| \Phi(x_t)} \leq \left(1-\delta\right)\Phi(x_t)$. Let
$\mathcal{T}$ be the first point in time where all nodes agree on color \A,
that is, $\mathcal{T}=\min\left\{ t \geq 0\colon \Phi(x_t)=0\right\}$. We
derive from \autoref{thm:drift-theorem} with parameters $\delta$ and 
$k=5\log n$ that 
$\Probability{\mathcal{T} \geq 20 /\delta\cdot\ln n} \leq n^{-5}$, where we
used the Taylor series approximation for $\log(1-\delta)$. Since
$\tMETA3-\tMETA2 \geq 20 /\delta\ln n$, the claim follows.
\end{proof}

\subsection{Putting Everything Together: Proof of \autoref{thm:async}}

We use \autoref{prop:bit-propagation} (which builds on \autoref{prop:two-choices}) and \autoref{lem:endgame}
to show \autoref{thm:async}, which is restated as follows.
\begin{restate}[thm:async]
\theoremasync
\end{restate} 
\begin{proof}
The first part of the proof is analogous to Case 2 of the proof of the synchronous version, \autoref{thm:memory}.
By \autoref{prop:bit-propagation} we have 
\[ x_j(\tBP{2}) =  \frac{ \hat c_j(\tau_0)^2}{ x(\tp{BP}{1}) }  \cdot \left(1 \pm \LittleO{1}\right)  + \BigO{n^{1-{4}/{\log\log n}}}  \enspace . \]
Observe that due to the definition of $x_j$ and \bulk, we have $x_j(\tFIN)=x_j(\tBP{2})$.
Furthermore, note that $\hat c_1(\tau_0) \geq n/k \geq n^{1-1/\log\log n}$ and hence
\[ 
\frac{\hat c_1(\tau_0)^2}{ x(\tp{BP}{1}) }  \geq n^{1-2/\log\log n} = \omega \left(   n^{1-{4}/{\log\log n}}     \right)
\]
Let $a':= \hat c_1(\tau_0+T)$ the nodes of color \A belonging to \bulk at the 
the beginning of the next round.
Define $b'$ analogously for color \B.
We consider 
the ratio between and show a quadratic growth w.r.t.\ $\hat c_1(\tau_0)^2/\hat c_2(\tau_0+T)^2$. 
We derive
\begin{equation*}
\frac{a'}{b'} \geq \frac{\frac{\hat c_1(\tau_0)^2}{ x(\tp{BP}{1}) } \cdot \left(1 - o(1) \right) }{\frac{\hat c_2(\tau_0)^2}{x(\tTC{})} \cdot \left(1 + o(1)) + \BigO{n^{1-{4}/{\log\log n}}}    \right)}
\geq \frac{\hat c_1(\tau_0)^2}{\hat c_2(\tau_0)^2} \cdot \left(1 - \LittleO{1}\right).
\end{equation*}
Hence, for sufficiently large constant $\ell$, we have after
$\ell\cdot\log\log n$ phases
\begin{equation}\label{eq:firstphaseisnice}
\hat c_1 \geq 19n/20 \enspace .
\end{equation}  

As mentioned before (see \autoref{obs:max-number-of-ticks}), using Chernoff
bounds, we can show that \whp:
\begin{enumerate}
\item All nodes have a working time in $[\tMETA0, \tMETA1)$ at reference point
$\tMETA0'$. This implies that no node starts with two choices phase before all
nodes finished Part 1 (\autoref{alg:asynchronous}).
\item All nodes have a working time in $[\tMETA0,\tMETA4]$ during the reference
points in $[\tMETA0', \tMETA3]$. This together with above statement and
\eqref{eq:firstphaseisnice} are the assumptions of \autoref{lem:neverback}.
\item All nodes have a working time in $[\tMETA1,\tMETA4]$ during the reference
points in $[\tMETA2, \tMETA3]$. This is the assumption required by
\autoref{lem:endgame}.
\end{enumerate}
Thus, by \autoref{lem:neverback} and \autoref{lem:endgame}, \whp all nodes agree
 on $\A$ at \tMETA3. Clearly, no node can change to any other color
afterwards and, by Chernoff bounds, after additional $\Theta(\log n)$ periods
all nodes will have completed the execution of \autoref{alg:endgame}. Thus the
total run time is $\Theta(\log n)$.
\end{proof}

\subsection{Increasing the Number of Opinions}
In our proofs we considered for the ease of presentation the setting where
$k\leq \Exp{\log n/\log\log n}$. 

However, it is possible to allow for any $k=\BigO{n^\varepsilon}$ (we still require that $a\geq (1+\varepsilon)b$).
This requires the algorithm to have a bound on $k$ so that
the length of block $\Delta$ is adapted to $\Delta=\Theta{(\log k + \log n/\log \log n)}$.
This is sufficient to get an equivalent notion of weak synchronicity. 
Due to the quadratic doubling, the algorithm requires  $\BigO{\log \log n}$ phases.
The length of the second part of the algorithm remains untouched resulting in 
a run time of   $\BigO{\log k \cdot \log \log n + \log n}$.

\section{Conclusion}
We introduced an algorithm to solve the plurality consensus in the asynchronous
setting. Our algorithm achieves the best the possible asymptotic run time in
the setting where the number of opinions $k$ is bounded by $\Exp{\log
n/\log\log n}$. We believe that the concept of weak synchronicity (including
the \SG and the tactical waiting) as well as our analysis techniques may well
prove to be of independent interest. Moreover, we feel that the ideas presented
here may be applicable to the adaptation of synchronous protocols to
asynchronous settings for a much wider class of problems, perhaps even
eventually leading to a generic framework. It  remains an open question whether
their exists an algorithm with the same run time allowing for
$k=\BigO{n^\varepsilon}$ opinions; we note that even in the synchronous setting
this questions remains open.

\section*{Acknowledgement.}
We would like to thank Gregor Bankhamer for helpful discussions and important hints. 

\nocite{*}
\clearpage
\bibliographystyle{abbrv}
\bibliography{paper}

\end{document}